\documentclass{article}
\usepackage{amsmath}
\usepackage{amsthm}
\usepackage{amsfonts}
\usepackage{cite}
\usepackage[T1]{fontenc}
\usepackage{libertine}
\usepackage[libertine]{newtxmath}
\usepackage[a4paper, margin=2.5cm]{geometry}
\usepackage{authblk}

\usepackage{hyperref}

\usepackage{thmtools, thm-restate}
\declaretheorem{theorem}
\declaretheorem{lemma}
\declaretheorem{claim}
\declaretheorem{definition}
\usepackage{booktabs}   %% For formal tables:
\usepackage{subcaption} %% For complex figures with subfigures/subcaptions
\usepackage{todonotes}

\usepackage{tcolorbox}
\usepackage{xspace}
\usepackage[ruled,vlined]{algorithm2e}

\usepackage{enumitem}

\newcommand{\R}{\mathbb{R}}

\DeclareMathOperator{\val}{val}

\DeclareMathOperator{\PG}{PG}

\DeclareMathOperator{\POLY}{poly}
\DeclareMathOperator{\low}{lower}
\DeclareMathOperator{\up}{upper}
\newcommand{\abs}[1]{|#1|}

\newcommand{\PSPACE}{\ensuremath{\texttt{PSPACE}}\xspace}
\newcommand{\NLOGSPACE}{\ensuremath{\texttt{NL}}\xspace}
\renewcommand{\P}{\ensuremath{\texttt{P}}\xspace}
\newcommand{\NP}{\ensuremath{\texttt{NP}}\xspace}
\newcommand{\coNP}{\ensuremath{\texttt{coNP}}\xspace}
\newcommand{\EXP}{\ensuremath{\texttt{EXP}}\xspace}
\newcommand{\TWOEXP}{\ensuremath{\texttt{2EXP}}\xspace}

\usepackage{pgfplots}
\usetikzlibrary{intersections}
\usetikzlibrary{calc}

\begin{document}

%%%%%%%%%%%%%%%%%%%%%%%%%%%%%%%%%%%%%%%%%%%%%%%%%%%%%%%%%%%%%%%%%%%%%%%%%%%%%%%
\title{One-Clock Priced Timed Games are \PSPACE-hard}         
%%%%%%%%%%%%%%%%%%%%%%%%%%%%%%%%%%%%%%%%%%%%%%%%%%%%%%%%%%%%%%%%%%%%%%%%%%%%%%%
\author{John Fearnley}
\author{Rasmus Ibsen-Jensen}
\author{Rahul Savani}
\affil{Dept.\ of Computer Science, University of Liverpool}
%%%%%%%%%%%%%%%%%%%%%%%%%%%%%%%%%%%%%%%%%%%%%%%%%%%%%%%%%%%%%%%%%%%%%%%%%%%%%%%

\maketitle

\begin{abstract}
The main result of this paper is that computing the value of a one-clock priced timed game (OCPTG) is \PSPACE-hard. Along the way, we provide a family of OCPTGs that have an exponential number of event points. Both results hold even in very restricted classes of games such as DAGs with treewidth three. Finally, we provide a number of positive results, including polynomial-time algorithms for even more restricted classes of OCPTGs such as trees.
\end{abstract}

\section{Introduction}

In this paper, we study \emph{priced timed games} (PTG), which are two-player
zero-sum games that are played on a graph. 
The defining feature of PTGs is that the game is played over time, with 
players accumulating costs both for spending time waiting in states, and for using
edges. Ultimately, one of the players would like to reach a goal state while
minimizing the cost, while the opponent would like to prevent the goal state
from being reached, or if that is impossible, to maximize the cost of reaching
the goal state.

Priced timed games have been studied extensively in the literature, starting
with the work of La Torre, Mukhopadhyay, and Murano~\cite{LMM02} who first
studied games on DAGs, with the later paper of Bouyer, Cassez, Fleury, and
Larsen~\cite{BCFL05} being the first to study the concept on general graphs.
Since then, there has been a great deal of follow-up work on these games,
e.g.,~\cite{LMM02,ABM04,BCFL05,BCFL05b,BBR05,BBM06,BLMR06,B06,BBBR07,BCDFLL07,JT07,R11,HIM13,BGNMMT14,BGHLM15,BJM15},
including work on practical applications in, for example, embedded systems, and
also applications in other theoretical results. 

\paragraph{\bf One-clock priced timed games.}
In general, a PTG can have any number of \emph{clocks}, which all increase at
the same rate as time progresses, but which can be independently \emph{reset}
back to zero. The edges of the game can have \emph{guards}, which only allow the
edge to be used if the clock values satisfy the conditions of the guard.

In this paper, we focus on the case in which there is exactly one clock, and so
we study \emph{one-clock priced timed games} (OCPTG). It has
been shown that one-clock priced timed games always have a
value~\cite{BLMR06}, and moreover algorithms have been proposed
\cite{HIM13,R11,BLMR06} for computing the value of these games. The current
state of the art is the algorithm of Hansen, Ibsen-Jensen, and
Miltersen~\cite{HIM13}, who give an algorithm that runs in $O(m \cdot \POLY(n)
\cdot 12^n)$ time, where $m$ is the number of edges and $n$ is the number of
vertices. This gives an exponential-time upper bound for the problem.

It has remained open, however, whether the problem can be solved in polynomial
time. The running time of Hansen, Ibsen-Jensen, and Miltersen's
algorithm~\cite{HIM13} is polynomial in the number of \emph{event points} in the
game, which are the set of points at which the gradient of the value function
changes. They showed that all OCPTGs have at most $24m(n+1)12^n$ event points, which
directly leads to the running time of their algorithm mentioned above.
They conjectured that the number of event points in a OCPTG is actually bounded
by a polynomial~\cite{HIM13}, and if this conjecture were true, then their
algorithm would always terminate in polynomial time.

\subsection{Our contribution}

\paragraph{\bf Lower bounds.}

This paper shows that computing the value of a one-clock priced timed game is
very unlikely to be solvable in polynomial time, by showing that the problem is
actually \PSPACE-hard. We begin by constructing a family of examples that have
exponentially many event points. This explicitly disproves the conjecture of
Hansen, Ibsen-Jensen, and Miltersen. We then use those examples as the
foundation of our computational complexity reductions. We first show that the
problem is both \NP and \coNP-hard, and we then combine the techniques from both
those reductions to show hardness for the $k$-th level of the polynomial-time hierarchy,
for all $k$, and finally \PSPACE.

All of our lower bound constructions produce graphs with special structures. In
particular they are all acyclic, planar, have in-degree and out-degree at most
$2$, and overall degree at most $3$ (our figures show a simpler variant with
overall degree at most 4). Also, the treewidth, cliquewidth, and rankwidth of
our constructions are all $3$.

Our results for the polynomial-time hierarchy give additional properties.
To obtain hardness for the $k$-th level of the polynomial time hierarchy, we
only need $k+2$ distinct \emph{holding rates}, which are the costs that the
players incur by waiting in a particular state.  

Another interesting feature is that, in a variant of the construction, which
loses planarity, all but $k+1$ of the states, can be made \emph{urgent}. A state
is urgent if the player is not allowed to wait at the state. Urgent states are
relevant because the results of~\cite{BLMR06,R11} are based on the technique of
converting more and more states into urgent states, since it is easy to solve a
game in which all states are urgent. 

In particular, our \NP- and \coNP-hardness constructions have 3 distinct holding
rates, namely $0,1/2,1$, and there is a variant in which all but two states can
be made urgent. 

Finally, members of our initial family that has exponentially many event points
have the additional properties of having pathwidth 3, using only holding rates
$\{0,1\}$, and having only a single state that cannot be made urgent. Thus, the
games may still have exponentially-many event points even for many of the most
obvious special cases.

%Our \NP and \coNP-hardness results are interesting in their own right, because
%they use a very small number of distinct \emph{holding rates}, which are the
%costs that the players incur by waiting in a particular state. In fact, both
%constructions use only the holding rates $\{0, 1/2, 1\}$, meaning that the
%problems are hard even with only three distinct holding rates. Another
%interesting feature is that, in these constructions, all but one of the states
%are  \emph{urgent}, meaning that the player is not allowed to wait at the
%state. Urgent states are relevant because the results of~\cite{BLMR06,R11} are
%based on the technique of converting more and more states into urgent states,
%since it is easy to solve a game in which all states are urgent. In our games,
%only one state needs to be converted, but it is \NP-hard or \coNP-hard to do the
%conversion.

%More generally, we obtain a fine-grained picture: we can construct families
%of games that are hard for the $k$-th level of the polynomial time hierarchy for
%any $k$.
%These games use $k+2$ distinct holding rates, and all but $k$ of the states are
%urgent. 
%\todo[inline]{Mention this somewhere in the main body.}

%Our construction produces graphs that have very special graph structures. The
%games are acyclic, planar, have in-degree and out-degree at most $2$, and
%overall degree at most $3$. For the \PSPACE-hardness result, the treewidth,
%cliquewidth, and rankwidth of the graphs are all bounded by $3$. 

\paragraph{\bf Upper bounds.}

Our hardness results essentially rule out finding polynomial-time algorithms for
many questions in a large number of special cases, unless \P=\xspace\PSPACE. We
are able to prove some upper bounds: we show that undirected graphs and trees
have a polynomial number of event points, and so can be solved in polynomial
time. 

Finally, we show that OCPTGs on DAGs are in \PSPACE by showing that a variant of
the event-point iteration algorithm~\cite{HIM13} can solve games on DAGS in
polynomial space. Combined with our hardness results, we obtain a
\PSPACE-completeness result for OCPTGs played on DAGs. This result improves on
an exponential-time algorithm by~\cite{ABM04} that in turn improved on a double-
exponential-time algorithm~\cite{LMM02}, both of which are designed for games
with many clocks. 

\subsection{Related work}

As shown in~\cite{BBM06}, building on a similar result in~\cite{BBR05} for five
clocks, some problems are undecidable in general for priced timed games with
three clocks. This was extended to the value problem in~\cite{BJM15}. The
complexity of most problems for two clocks is still open. 

Games with only a single player, called priced timed automata, have been studied
extensively on their own, following their introduction in~\cite{ALP01,BFHLPRV01}. 
They can be solved in \NLOGSPACE for the one-clock case~\cite{LMS04} and in
\PSPACE for the multiple clock case~\cite{BBBR07}. Games on DAGs are in \EXP for
any number of clocks~\cite{ABM04}, which improved on a previous \TWOEXP bound
in~\cite{LMM02}. Games with no costs and holding rates in $\{0,1\}$ are called
reachability timed games. They can be solved in polynomial time for one
clock~\cite{T11,HIM13} and in polynomial space for multiple clocks~\cite{JT07}.

This result has been generalised by~\cite{BGNMMT14} to show a polynomial time algorithm for the decision question\footnote{I.e. given a game, a state, and a number, is the value of starting in that state at time~0  above the value?} for one clock=priced timed games with rates $\{0,1\}$ and integer costs.
 Previously, they also claimed that such games would have only a polynomial number of event points, implying that one could find the full value functions in polynomial time. This,
however is incorrect: We show in Appendix~\ref{app:ashutosh} how to convert our examples
with exponentially-many event points and two holding rates to have integer costs. 
Their result~\cite{BGNMMT14} does show a pseudo-polynomial number of event points for such games though.

More generally, \cite{BGNMMT14} also give a pseudo-polynomial time algorithm for the special case with holding rates in  $\{-d,0,d\}$ for any number $d$ (note that our paper otherwise does not discuss negative rates or costs).

\section{Definitions}

As shown by~\cite{HIM13}, every one-clock priced timed game can be reduced, in
polynomial time, to a \emph{simple} priced timed game (SPTG), which is an OCPTG
in which there are no edge guards and no clock resets. Our hardness results will
directly build SPTGs, and so we restrict our definitions to SPTGs in this
section. Since every SPTG is a OCPTG, all of our hardness results directly
apply to OCPTGs. 

%are polynomial time reducible to the special case of
%simple priced timed games (SPTGs), where all intervals are $[0,1]$ and the clock
%is never reset. This paper is therefore focusing on SPTGs.

%In this section, we will give formal definitions of the different types of games considered. Note that the definition of OCPTGs is mainly for self-containment and we will focus on SPTGs (because they are polynomial time equivalent~\cite{HIM13}).

\paragraph{\bf SPTGs.} 

A simple priced timed game is a game played between two players called the
minimizer and the maximizer. The game is formally defined by a $6$-tuple: 
$(V_1,V_2,G,E,c,r)$, where

\begin{itemize} [leftmargin=0.3cm]

\item $V_1$ is the set of states belonging to the minimizer, $V_2$ is the set of
states belonging to the maximizer, and $G$ is a set of goal states.
%$V_1,V_2,G$ are the disjoint {\em player~1}, {\em player~2}, {\em goal}
%\footnote{one could use $G$ as a singleton, but we use this definition because it matters for which special graph properties the graphs can have} 
%{\em states}, respectively. 
The set of all states is denoted as $V=V_1\cup V_2\cup G$, and we use $n$ to
denote the number of states.

\item $E$ is a set of directed {\em edges}, which is a subset of $V\times V$. We
use $m$ to denote the number of edges.  

\item $c:E\rightarrow \R_{\geq 0}$ is a non-negative {\em cost function} for
edges.

\item $r:V\rightarrow \R_{\geq 0}$
is a non-negative {\em holding rate function} for states.

%\item $M$ is the {\em maximum time}.

%\item $I$ is the {\em existence interval function} for edges, i.e. $I:E\rightarrow \I(M)$, where $\I(M)$ is the set of intervals that are subsets of $[0,M]$.

%\item $R$ is the set of directed {\em reset edges}, i.e. a subset of $E$.

\end{itemize}

The game takes place over a period of time. At the start of the game, a pebble
is placed on one of the states of the game. In each round of the game, we will
be at some time $t \in [0, 1]$. The player who owns the state that holds
the pebble, can choose to move the pebble along one of the outgoing edges of
that state, or to delay until some future point in time. Moving along an edge
$e$ incurs the fixed one-time cost given by $c(e)$, while delaying for $d$ time
units at a state $s$ incurs a cost of $r(s) \cdot d$. 

The game starts at time $0$, and either ends when a goal state is reached, or
it never ends. If a goal state is not reached, then the minimizer loses the game, and receives payoff $-\infty$.
Otherwise, the payoff is the total amount of cost that was incurred before the
goal state was reached, which the maximizer wins, and the minimizer loses.

\paragraph{\bf Strategies.} 

Our players will use \emph{time-positional strategies}, meaning that for each
state and each point in time, the strategy chooses a fixed action that is
executed irrespective of the history of the play.
Formally, for each $j\in \{1,2\}$, a time-positional strategy $\sigma_j$ for player $j$
is defined by a pair $(W^j,S^j)$.
\begin{itemize} [leftmargin=0.3cm]

\item $W^j$ is a set of non-negative {\em
change points}. That is, $W^j=\{0=w_0^j<
w_1^j<w_2^j<\dots<w_{k-1}^j<1=w_{k}^j\}$ gives a sequence of points in time at
which the player changes their strategy. For notational convenience we define $w_{k+1}^j=\infty$.
\item $S^j=\{S_0^j,S_1^j,\dots S_{k}^j\}$ is a corresponding list of \emph{strategy choices},
which defines what action the player chooses at each point in time.
The player can either choose an outgoing edge, or choose to wait at the state,
which we denote with the symbol $\delta$.
So, for each $i$, we have that $S_i^j:V_j\rightarrow E\cup \{\delta\}$ with the
requirement that if 
$S_i^j(s)\in E$ then $S_i^j(s)=(s,s')$ for some state $s'$. At time~$1$,
delay is not possible, so for all $s\in V_j$ we require that
$S_{k}^j(s) \ne \delta$. 
\end{itemize}

\paragraph{\bf Plays.} 

Given a pair of strategies $\sigma_1,\sigma_2$ for the minimizer and the
maximizer, respectively, the resulting {\em play} from a starting state $s_0$,
and a starting time $t_0$ is denoted as $P(\sigma_1,\sigma_2,s_0,t_0)$, and is
defined as follows. 
Initially,
place a pebble on $s_0$ at time $t_0$. For each $j\in \{1,2\}$ and $i$,
whenever the pebble is placed on a state $s_i$ in $V_j$ at time $t_i$, let $i'$
be the index such that $t_i\in [w^j_{i'},w^j_{i'+1})$ and let $\ell\geq i'$ be the
smallest index such that $e_i := S_{\ell}^j(s_i) = (s_i, s_{i+1}) \neq \delta$.
Then, player $j$ waits until time $t_{i+1} = w_\ell$, and then
moves the pebble on to $s_{i+1}$ at time $t_{i+1}$
We also define $\delta_{i}:=t_{i+1}-t_i$ to be the delay that player~$i$ chooses
at time $t_i$. 

If $s_i\in G$, then the play is over and $|P(\sigma_1,\sigma_2,s_0,t_0)|=i$. If
the play is never over, i.e. for all $i$, $s_i\not\in G$, we have that
$|P(\sigma_1,\sigma_2,s_0,t_0)|=\infty$.

\paragraph{\bf Outcomes and values.} 

The {\em outcome}
$\val(P)$ is defined to be~$\infty$ if $|P|=\infty$, since no goal state is
reached. 
Otherwise, the outcome is
\[\val(P) := \sum_{t=0}^{|P|}(r(s_t)\cdot \delta_{t}+c(e_t)),\] 
where $r(s_t)\cdot \delta_{t}$ is the cost for holding at the state $s_t$ for
$\delta_t$ time units, and $c(e_t)$ is the cost for using the edge $e_t$.
Fix $s$ to be a state, and $t$ to be a time. The
{\em lower value} is defined to be 
$\underline{\val}(s,t)=\sup_{\sigma_1}\inf_{\sigma_2}
\val\left(P(\sigma_1,\sigma_2,s,t)\right),$ while the {\em upper value} is defined to be
$\overline{\val}(s,t)=\inf_{\sigma_2} \sup_{\sigma_1}
\val\left(P(\sigma_1,\sigma_2,s,t)\right).$ 

By definition, $\underline{\val}(s,t)\leq \overline{\val}(s,t)$. As shown in
\cite{BLMR06}, for a richer class of strategies,
$\underline{\val}(s,t)=\overline{\val}(s,t)$. It mostly follows from
\cite{BLMR06} (but formally, one also needs \cite{HIM13}) that this equality
holds even when the minimizer is restricted to time-positional strategies in the
definition of lower value and the maximizer is restricted to time-positional
strategies in the definition of upper value.  Therefore, the game is
determined in time-positional strategies, and we use
$\val(s,t):=\underline{\val}(s,t)=\overline{\val}(s,t)$ to denote the value of
the game starting at the state $s$, and time $t$.
%and we do not need to consider more general strategies. 

\paragraph{\bf Optimal and $\epsilon$-optimal strategies.} 

Given an $\epsilon\geq 0$, a strategy $\sigma_1$ is
\emph{$\epsilon$-optimal} for the minimizer if $\val(s,t)-\epsilon\leq
\inf_{\sigma_2} \val(P(\sigma_1,\sigma_2,s,t))$ for all $s$ and $t$.  A
strategy is \emph{optimal} if it is $0$-optimal. The definitions for the
maximizer are symmetric.
As shown in \cite{BLMR06}, for all
$\epsilon>0$, $\epsilon$-optimal strategies exist in OCPTGs, and 
\cite{HIM13} have shown that optimal strategies exist in SPTGs.
Moreover, the function $\val(s,t)$
is piecewise linear for OCPTGs \cite{BLMR06}, and continuous for
SPTGs~\cite{HIM13}

\paragraph{\bf Event points.} 

As mentioned, the value function of each state in an SPTG is
piecewise linear. 
An \emph{event point} is a point in time at which the value function of some
state $s$ changes from one linear function to another.
The set of {\em event points} contains every event point for every
state in the game.

%is the set of points $E$ in time such that $t\in E$ iff there exists some state
%$s$ for which $\val(s,t)$ is not differentiable when viewed as a function of $t$
%(i.e., the event points are the ends of the line pieces making up the value
%functions). 
As shown in \cite{HIM13}, improving on \cite{R11,BLMR06}, the number of event
points is less than $12^n$ for SPTGs and it is less than $m\cdot 12^n\cdot
\POLY(n)$ for OCPTGs. The optimal strategies for SPTGs constructed
by~\cite{HIM13} have the set of change points being equal to the set of event
points. Conversely, it is clear that event points are a subset of the change
points in any optimal strategy.

\section{Exponentially many event points are required}
\label{sec:exponential}

We begin by constructing a family of simple priced timed games in which the
number of event points in the optimal strategy is exponential. This serves two
purposes. Firstly, it provides a negative answer to the question, posed in prior
work~\cite{HIM13}, of whether the number of event points is polynomial.
Secondly, this construction will be used in a fundamental way in the hardness
results that we present in later sections.

\paragraph{\bf The construction.}

\begin{figure*}
\input figures/lb_combined.tex
\caption{Event points lower bound construction.}
\label{fig:basiclb}
\end{figure*}

The family of games is shown on the left-hand side in Figure~\ref{fig:basiclb}.
States belonging to the maximizer are drawn as squares, while states belonging
to the minimizer are drawn as triangles. The number displayed on each state is
the holding rate for that state, while the number affixed to each edge is the
cost of using that edge.

The game is divided into levels, with each level containing two states, which we
will call the \emph{left state} (denoted as $v_{\ell}^i$) and the
\emph{right state} (denoted as $v_{r}^i$). These names correspond to the
positions at which these states are drawn in Figure~\ref{fig:basiclb}. 

At the bottom of the game, on level~0, the left state $v_{\ell}^0$ is the goal
state and the right state $v_r^0$ is a maximizer state with holding rate~1. The state
$v_r^0$ has an edge to $v_{\ell}^0$ with cost 0. For each level $i > 0$, the
left state $v_{\ell}^i$ is a minimizer state with holding rate~1, and the right state
$v_{r}^i$ is a maximizer state with holding rate~0. Both states have the same outgoing
edges: an edge to $v_{\ell}^{i-1}$ with cost $2^{-i}$, and an edge to
$v_{\ell}^{i-1}$ with cost $0$.

\paragraph{\bf Value diagrams.}

On the right-hand side in Figure~\ref{fig:basiclb}, we show the value function
for each state, represented as \emph{value diagrams}. These show the value for
each state at each point in time. The bottom-left diagram shows the value
function of the goal state, which is zero at all points in time, since the game
ends when the state is reached. The bottom-right diagram shows the value
function of the state $v_r^0$ (the bottom-right state of the game). At this
state, the maximizer will wait for as long as possible before moving to the
goal, since this maximizes the cost generated from the holding rate of~$1$.
Hence, the value of this state is $1 - x$ at time $0 \le x \le 1$, which is
shown in the diagram.

For the states at level one of the game, first observe that there is no
incentive for either player to wait. The left state has holding rate $1$, which
is the worst possible holding rate for the minimizer, and the right state has holding
rate $0$, which is the worst possible holding rate for the maximizer. Hence both players
will move immediately to the lower level, and we must determine which state is
chosen.

To do this, we use the value function diagrams of the lower level. Both players
can move to the goal with an edge cost of $0.5$, or move to $v_r^0$ with a cost of
zero. So we shift the value function of the goal state up by $0.5$, and then
overlay it with the value function of $v_r^0$. This is displayed in the value
diagram that lies between the two layers. The minimizer's value function is the
\emph{lower envelope} of these two functions, which minimizes the value, 
while the maximizer's value function is the
\emph{upper envelope}, which maximizes the value. 
This is shown in the value diagrams of the two states at level one.

This process repeats for each level. For level two, we overlay the two value
diagrams from level one, after shifting the left-hand diagram up by the edge
cost of $0.25$, and then we take lower and upper envelopes for the respective
players.

\paragraph{\bf The exponential lower bound.}

To see that this game produces exponentially many event points, observe that the
left-hand value diagram at level two contains two complete copies of the
left-hand value diagram at level one, and that the same property holds for the
right-hand value diagrams. This property generalizes, and we can show that the
value diagrams for $v_{\ell}^{n}$ and $v_r^n$ both contain $2^n$ distinct line
segments. The following theorem is shown in Appendix~\ref{app:exp_event_points}.

%\begin{theorem}\label{thm:exp_event}
%There is a family of simple priced time games that have exponentially many event
%points.
%\end{theorem}

% http://texdoc.net/texmf-dist/doc/latex/thmtools/thmtools.pdf
\begin{restatable}{theorem}{expeventthm}
\label{thm:exp_event}%
There is a family of simple priced time games that have exponentially many event
points.
\end{restatable}

\subsection{Inapproximability with few change points}
\label{sec:no_approx}

We are also able to show that both players must use strategies
with exponentially many change points in order to play close to optimally
in our lower bound game.
Specifically, we can show that if the game starts at the $k$th level of our
game, that is, in the vertices 
$v_{\ell}^k$ or $v_r^k$, and if both players play $\epsilon$-optimally for 
$\epsilon < 1/2^{k}$, then every interval of the form
\begin{equation*}
\left[ \frac{x}{2^{k-1}}, \frac{x+1} {2^{k-1}} \right)
\end{equation*}
for some integer $x$, must contain a change point. This is only possible if
there are $2^{k-1}$ distinct change points.

We shall illustrate this for the case where $k = 3$, by showing that the
minimizer must use four different change points at $v_{\ell}^3$ to play an
$\epsilon$-optimal strategy with $\epsilon < 1/8$. The value
diagram of $v_{\ell}^3$ is the lower envelope of the value diagram at the top of
Figure~\ref{fig:basiclb}. Let us consider the interval $D=\left[x/4,(x+1)/4\right)$ for
some integer $x$, and for the sake of contradiction, suppose that there are no
change points in this interval.

Since the minimizer cannot change their strategy, they have only three options
during $D$: always go to $v_{\ell}^2$, always go to $v_{r}^2$, or wait at
$v_{\ell}^3$ until the end of $D$. 

If the minimizer chooses to wait, then let us consider a play starting at time
$x/4$. This play has a payoff of at least $$1/4 + \val(v_{\ell}^3,(x+1)/4),$$
because we wait with a holding rate of $1$ for $1/4$ time units, and then the
best we can do at time $(x+1)/4$ is to follow the optimal strategy, which gives
us a payoff of $\val(v_{\ell}^3,(x+1)/4)$. On the other hand, we have
$$\val(v_{\ell}^3,x/4)=\val(v_{\ell}^3,(x+1)/4)+1/8.$$ This can be seen from the
value function for $v_{\ell}^3$ in Figure~\ref{fig:basiclb}: the first half of
the value function during $D$ is flat, while the second half falls at  rate $1$,
hence the difference is $1/8$. Since choosing to wait achieves a value that is
$1/8$ bigger than this, waiting cannot be $\epsilon$-optimal for any $\epsilon <
1/8$.

For the other two options, the outcomes can be seen in the top value diagram in
Figure~\ref{fig:basiclb}. The red line gives the outcome for starting in
$v_{\ell}^3$ and always going to $v_{r}^2$, while the blue line gives the
outcome for always going to $v_{\ell}^2$, assuming that both players play
optimally afterwards. The optimal strategy takes the lower envelope of the two
lines.

There is a difference of $1/8$ between the two lines at $x/4$ and $(x+1)/4$, but
the lines cross in the middle of the interval, so the line that is part of the
lower envelope at $x/4$ is not the line that is part of the lower
envelop at $(x+1)/4$. Hence, choosing to go to $v_{r}^2$, or to $v_{\ell}^2$ for
the entire interval will cause a loss in value of up to $1/8$, relative to the
optimal strategy, which is the difference in height between the lines. The
strategy is therefore not $\epsilon$-optimal since $\epsilon<1/8$.

This argument is generalized to all $k \geq 1$ and both players in
Appendix~\ref{app:no_approx}. 
Ultimately, the result is stated in the following lemma.
%\begin{lemma}\label{lem:no_approx}
%There is a family of simple priced time games in which every $\epsilon$-optimal
%strategy with $\epsilon < 1/2^k$ uses $2^{k-1}$ change points.
%\end{lemma}
\begin{restatable}{lemma}{noapproxlem}
\label{lem:no_approx}
There is a family of simple priced time games in which every $\epsilon$-optimal
strategy with $\epsilon < 1/2^k$ uses $2^{k-1}$ change points.
\end{restatable}

\section{\NP and \coNP lower bounds}
\label{sec:np_conp}

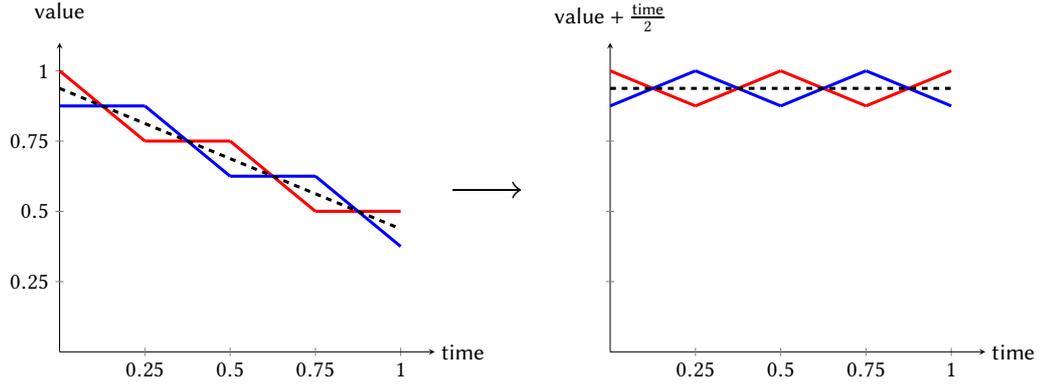
\begin{figure}
\centering
\scalebox{1.8}{
% https://tex.stackexchange.com/questions/23411/tikzpicture-in-node-of-another-tikzpicture-how-to-screen-of-from-inheriting-sty
%%%%%%%%%%%%%%%%%%%%%%%%%%%%%%%%%%%%%%%%%%%%%%%%%%%%%%%%%%%%%%%%%%%%%%%%%%%%
 
\pgfkeys{/pgfplots/mynewaxis/.style={xmin=0,ymin=0,xmax=1.1,ymax=1.1, 
								  samples=2, % straight lines so no need for more than 2
								  axis lines=center,
						  		  xtick = {1},
						          ytick = {1},
								  ylabel=$\mathsf{value}$,
								  xlabel=$\mathsf{time}$,
								  x label style = {at={(ticklabel* cs:1.15)}, anchor=east, font=\large},
								  y label style = {at={(ticklabel* cs:1.15)}, anchor=north, font=\large},
								  tick label style={font=\large}
						  		  }}

\pgfkeys{/pgfplots/mynewaxis2/.style={xmin=0,ymin=0,xmax=1.1,ymax=1.1, 
								  samples=2, % straight lines so no need for more than 2
								  axis lines=center,
						  		  xtick = {1},
								  ytick = \empty,
								  yticklabel=\empty,
								  ylabel=$\mathsf{value} + \frac{\mathsf{time}}{2}$,
								  xlabel=$\mathsf{time}$,
								  x label style = {at={(ticklabel* cs:1.15)}, anchor=east, font=\large},
								  y label style = {at={(ticklabel* cs:1.15)}, anchor=north, font=\large},
								  tick label style={font=\large}
						  		  }}

\pgfmathsetmacro{\scaling}{0.4}

% combinequarter copied from lb_combined.tex
\newsavebox\torotate
\begin{lrbox}{\torotate}
\scalebox{\scaling}{ 
\begin{tikzpicture}
\begin{axis}[mynewaxis, extra x ticks={0.25, 0.5, 0.75}, 
					 extra y ticks={0.25, 0.5, 0.75}]
  % old right part
  \addplot[red, ultra thick][domain=0:0.25] (x,1-x);
  \addplot[red,  ultra thick][domain=0.75:1] (x,0.5);
  \addplot[red,  ultra thick][domain=0.25:0.5] (x,0.75);
  \addplot[red, ultra thick][domain=0.5:0.75] (x,1.25-x);
  % old left part shifted up by 0.125
  \addplot[blue, ultra thick][domain=0.25:0.5] (x,1.125-x);
  \addplot[blue,  ultra thick][domain=0.5:0.75] (x,0.625);
  \addplot[blue,  ultra thick][domain=0:0.25] (x,0.875);
  \addplot[blue, ultra thick][domain=0.75:1] (x,1.375-x);
  % the "rotation" line
  \addplot[ultra thick, dashed][domain=0:1] (x,0.9375-0.5*x);

\end{axis}
\end{tikzpicture}
}
\end{lrbox}

\newsavebox\rotated
\begin{lrbox}{\rotated}
\scalebox{\scaling}{ 
%\begin{tikzpicture}[rotate=32,transform shape]
\begin{tikzpicture}
%\begin{axis}[mynewaxis2]
\begin{axis}[mynewaxis2, extra x ticks={0.25, 0.5, 0.75}, 
					 extra y ticks={0.25, 0.5, 0.75}]
%\begin{axis}[axis lines=none]
  % old right part
  \addplot[red, ultra thick][domain=0:0.25] (x,1-x+0.5*x);
  \addplot[red,  ultra thick][domain=0.75:1] (x,0.5+0.5*x);
  \addplot[red,  ultra thick][domain=0.25:0.5] (x,0.75+0.5*x);
  \addplot[red, ultra thick][domain=0.5:0.75] (x,1.25-x+0.5*x);
  % old left part shifted up by 0.125
  \addplot[blue, ultra thick][domain=0.25:0.5] (x,1.125-x+0.5*x);
  \addplot[blue,  ultra thick][domain=0.5:0.75] (x,0.625+0.5*x);
  \addplot[blue,  ultra thick][domain=0:0.25] (x,0.875+0.5*x);
  \addplot[blue, ultra thick][domain=0.75:1] (x,1.375-x+0.5*x);
  % the "rotation" line
  \addplot[ultra thick, dashed][domain=0:1] (x,0.9375-0.5*x+0.5*x);

\end{axis}
\end{tikzpicture}
}
\end{lrbox}

% add the following two lines to your document to get bigger arrows
%\usetikzlibrary{arrows.meta}
%\tikzset{>={Latex[width=2mm,length=2mm]}}

\begin{tikzpicture}[auto,scale=1]%,every node/.style={scale=1.3}]

\draw (0, 0) node {\usebox\torotate};
\draw[->] (1.5, 0) -- (2, 0);
\draw (4, 0) node {\usebox\rotated};

\end{tikzpicture} 
}
\caption{Our relative value diagramming convention.}
\label{fig:rotate}
\end{figure}

\begin{figure}
\centering
\scalebox{1}{
\input figures/np.tex
}
\caption{\NP lower bound construction.}
\label{fig:np}
\end{figure}

We now present \NP-hardness and \coNP-hardness results for computing the value of a simple priced
timed game. This serves two purposes. Firstly, it introduces some of the key
concepts that we will use in our \PSPACE-hardness result. Secondly, these 
hardness results will hold for SPTGs that have only the holding
rates $\{0,1/2,1\}$, which is not the case for our later results.

Our goal in this section is to show hardness results for the following decision
problem: given a state $s$ and a constant~$c$, decide whether $\val(s, 0) \ge
c$. 
In other words, it is hard to determine the value of a particular state
at time zero. The majority of this section will be used to describe the
\NP-hardness result, and the \coNP-hardness will be derived by slightly altering
the techniques that we develop.

\paragraph{\bf Relative values.}

The family of games from Section~\ref{sec:exponential} will be used as a basis
for this result. 
We start by discussing a change in perspective that is helpful when
dealing with value diagrams. Take, for example, the value diagram at the top of
Figure~\ref{fig:basiclb}. Observe that both of the value functions depicted in
this diagram are weakly monotone. This will always be the case in an SPTG, since
there are no guards, meaning that costs can only increase as the amount of time
left in the game increases.

We will use values at specific points in time to encode information. But we will
not use the absolute value, but rather the value \emph{relative} to some
monotone linear function. This is shown in Figure~\ref{fig:rotate}. 
On the right-hand side we have added the linear function $\text{time}/2$, which
causes the value functions to become horizontal. The diagram shows the value
functions increasing and decreasing relative to this linear function.

We will use relative values in our reduction, because it makes it easier to
understand.  It is worth pointing out, however, that this is only a change in
perspective. The underlying absolute values are still always weakly monotone.

\paragraph{\bf Enumerating bit strings.}

Our \NP-hardness reduction will be a direct reduction from Boolean
satisfiability. There are two steps to the reduction. First we build a set of
gadgets that enumerate all possible $n$-bit strings over time,
%(i.e. at any one point in time, our gadgets corresponds to one $n$-bit string
%and each $n$-bit string has its time) 
and then we build a gadget that tests whether a Boolean formula is true over
this set of bit strings.

We start by describing the enumeration gadget. We denote the $n$ bits of a bit
string as $v_1$ through $v_n$. The enumeration gadget builds $2n$ states,
corresponding to $v_i$ and $\lnot v_i$ for each index $i$.
The top half of Figure~\ref{fig:np} shows the relative value diagrams for
these states. 

The gadget divides time into $2^n$ intervals, with each interval corresponding
to a particular bit string. Bit values of the bit-string are encoded using the
relative value function of the states, using two fixed constants $L$ and $H$ 
that the relative value stays between.
\begin{itemize} [leftmargin=0.3cm]

\item If a bit is zero for an interval, then the relative value
of the state remains at $L$ during the interval.
\item If a bit is one for an interval, then the relative value of the state
begins the interval at $L$, it increases during the interval to $H$, and then
decreases back to $L$ by the end of the interval. This forms the peaks 
shown in Figure~\ref{fig:np}.
\end{itemize}

The enumeration gadget produces these value functions by using several copies of
the exponentially-many event point games from Section~\ref{sec:exponential}. From
Figure~\ref{fig:basiclb}, we can see that the value functions there are similar
to what we want: the functions alternate between having high relative value and
low relative value, and there are exactly $2^i$ alternations at level $i$.
However, these value functions do not exactly match those shown in
Figure~\ref{fig:np}. Specifically:
\begin{itemize} [leftmargin=0.3cm]

\item The exponential lower bound functions are symmetric with respect to peaks
and troughs, but we would like zeroes to be represented by the fixed constant
$L$, and ones to be represented as peaks. 
\item The functions start at either peaks or troughs, but we would
like to start in the middle of the waveform. So attempting to represent $v_1$ in
Figure~\ref{fig:np} using the value functions from the exponential lower
bound would result in a bit-sequence like
\verb+1 1 0 0 0 0 1 1+, rather than \verb+0 0 0 0 1 1 1 1+.
\item When a state has a sequence of intervals that all encode one-bits, we
would like each to contain a copy of the peaks shown in
Figure~\ref{fig:np}. However, the exponentially-many event point game value
functions would instead give us a single large peak during the whole interval.  
\end{itemize}

To address these issues, we transform the exponentially-many event point game 
value functions so that they have these properties. This involves inserting a sequence
of intermediate states, and the construction is described in detail in
Appendix~\ref{app:enc_bools}.

\paragraph{\bf Evaluating a Boolean formula.}

Once we have constructed the states $v_1$ through $v_n$ and $\lnot v_1$ through
$\lnot v_n$, we can then design a gadget to evaluate an arbitrary Boolean
formula $F$ over every $n$-bit string. The output of this gadget is a state, that
we will also call $F$, whose value is depicted in Figure~\ref{fig:np}. 
The output of the $F$ state uses the same encoding as before: if $F$ evaluates
to false for a specific bit string, then the value of $F$ remains at $L$ for the
entire interval, while if it evaluates to true, the value forms a peak that
starts at $L$, increases to touch $H$, and then returns to $L$ by the end of the
interval.

To evaluate the formula, we first apply De Morgan's laws to ensure that all
negations are applied to propositions, meaning that all internal operations of
the formula consist only of $\land$ and $\lor$ operations.
Next, we introduce a state in the game for each sub-formula $F' = x \oplus y$ of
$F$, where $\oplus \in \{\land, \lor\}$. This state will have edges to the
states corresponding to $x$ and $y$ with no edge costs, and 
\begin{itemize} [leftmargin=0.3cm]

\item 
if $\oplus = \lor$ then the state is a \emph{maximizer} state with
\emph{holding rate $0$}, while
\item 
if $\oplus = \land$ then the state is a \emph{minimizer} state with
\emph{holding rate $1$}.
\end{itemize}

As in the exponentially-many event point games, the holding rates have been chosen so that
neither player has an incentive to wait at these states. So the relative value
of the state $F'$ 
\begin{itemize} [leftmargin=0.3cm]

\item will be the maximum of the two input states for an
$\lor$ gate, meaning that in any particular interval the relative value of $F'$
will contain a peak if either of the two input states contains a peak,
\item 
while for a $\land$ gate, the relative value will be the minimum of the two inputs,
meaning that an interval will contain a peak only when both inputs contain
peaks\footnote{An issue could arise if the peaks were located at different
points in the intervals, but as shown in Lemma~\ref{lem:detector} of
Appendix~\ref{sec:form_enc}, the peaks are always exactly in the middle.}.
\end{itemize}
Hence this correctly simulates boolean logic, and the output of state $F$ will
encode the set of bit strings that satisfy the formula.

\paragraph{\bf \NP-hardness of computing values.}

Finally, we can turn this into our \NP-hardness result. So far, we have shown
how to evaluate the Boolean formula, but the outcome of the evaluation does not
affect the values at time zero, because each evaluation is entirely contained
within its interval. 

To address this, we introduce one final state called the \emph{extender}. 
The relative value function of the extender is shown at the bottom of
Figure~\ref{fig:np}. Whenever the relative value of $F$ peaks at the value
$H$, the extender makes the relative value decay more gradually on the left-hand
side of the peak. This decay rate is carefully chosen so that the value will not
have returned to $L$ even after all $2^n$ intervals. Hence, 
\begin{itemize} [leftmargin=0.3cm]

\item if the relative value of $F$
touches $H$ at any point in time, the relative value of the extender at time
zero
will be strictly greater than $L$, while 
\item if the relative value of
$F$ is never more than $L$, then the relative value of the extender will be $L$
at time zero.
\end{itemize}
This implies that the relative value (and hence absolute value) of the extender
at time zero depends on the satisfiability of the formula $F$, which gives us
our \NP-hardness result.

The extender state is a maximizer state that has one outgoing edge to the state
$F$ with no edge cost, and a carefully chosen holding rate. The second to last
relative value diagram in Figure~\ref{fig:np} shows the affect of the holding
rate of the extender. The idea is that the maximizer would like to wait in the
extender until the next interval in which the formula evaluates to true (if
there is such an interval). 

The holding rate at the extender determines the gradient of the blue lines. For
\NP-hardness it is sufficient for this line to be horizontal\footnote{Horizontal 
in our relative value diagrams means a holding rate of $1/2$ in the actual game
with absolute values.},
and never touch the relative value of $L$, but the ability for the extender
state to decay back to $L$ after a finite amount of time will be later used in
our \PSPACE-hardness result.

One final thing to note is that this construction uses exactly three different
holding rates. The exponentially-many event point games use the holding rates
$0$ and $1$, and one extra holding rate (of $1/2$) is introduced in the
enumeration gadget. 
We get the following theorem. 
The full formal description of the construction, along with a proof of
correctness, can be found in Appendix~\ref{app:np_conp_hardness}. 
%\begin{theorem}\label{thm:np_hard}
%For an SPTG, deciding whether $v(s, 0) \ge c$ for a given state $s$ and constant
%$c$ is \NP-hard, even if the game has only holding rates in $\{0,1/2,1\}$.
%\end{theorem}
\begin{restatable}{theorem}{nphardthm}
\label{thm:np_hard}
For an SPTG, deciding whether $v(s, 0) \ge c$ for a given state $s$ and constant
$c$ is \NP-hard, even if the game has only holding rates in $\{0,1/2,1\}$.
\end{restatable}

\paragraph{\bf \coNP-hardness of computing values.}

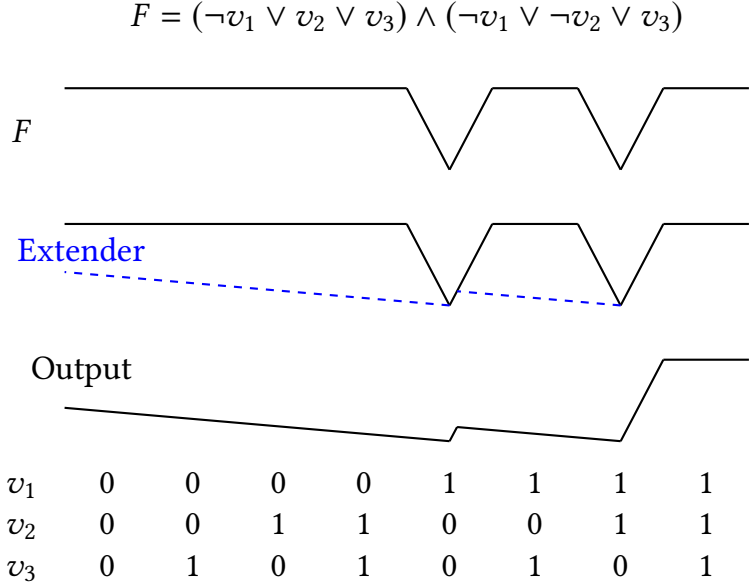
\begin{figure}
\centering
\scalebox{1}{
\begin{tikzpicture}[scale=2.25,every node/.style={scale=1.3}]

%VARIABLES
\pgfmathsetmacro{\width}{2}

%% original
%\pgfmathsetmacro{\spikeheight}{1}
%\pgfmathsetmacro{\panelheight}{1.5}

% original
\pgfmathsetmacro{\spikeheight}{0.48}
\pgfmathsetmacro{\panelheight}{0.8}
\pgfmathsetmacro{\middlegap}{0.5}
\pgfmathsetmacro{\leftlabel}{-0.5}

\pgfmathsetmacro{\levelzerobottom}{0}
\pgfmathsetmacro{\levelzero}{\levelzerobottom + \spikeheight}

\pgfmathsetmacro{\levelonebottom}{-1*\panelheight}
\pgfmathsetmacro{\levelone}{\levelonebottom + \spikeheight}

\pgfmathsetmacro{\leveltwobottom}{-2*\panelheight}
\pgfmathsetmacro{\leveltwo}{\leveltwobottom + \spikeheight}

%\pgfmathsetmacro{\beneath}{-3.5}
\pgfmathsetmacro{\beneath}{-2*\panelheight-0.5}
\pgfmathsetmacro{\above}{1.5}

\node at (1.75, 0.9) {$F = (\lnot v_1 \lor v_2 \lor v_3) \land (\lnot v_1 \lor \lnot v_2 \lor v_3)$};

%%%%%%%%%%%%%%%%%%%%%%%%%%%%%%%%%%%%%%%%%%%%%%%%%%%%%%%%%%%%%%%%%%%%%%%%%%%%%%
% KEY AT bottom
%%%%%%%%%%%%%%%%%%%%%%%%%%%%%%%%%%%%%%%%%%%%%%%%%%%%%%%%%%%%%%%%%%%%%%%%%%%%%%
%\node at ( -0.5,\above) {$\begin{matrix} v_1 \\ v_2 \\ v_3\\ \end{matrix}$};
%\node at (  0,\above) {$\begin{matrix}  0 \\ 0 \\ 0 \\ \end{matrix}$};
%\node at (0.5,\above) {$\begin{matrix}  0 \\ 0 \\ 1 \\ \end{matrix}$};
%\node at (  1,\above) {$\begin{matrix}  0 \\ 1 \\ 0 \\ \end{matrix}$};
%\node at (1.5,\above) {$\begin{matrix}  0 \\ 1 \\ 1 \\ \end{matrix}$};
%\node at (  2,\above) {$\begin{matrix}  1 \\ 0 \\ 0 \\ \end{matrix}$};
%\node at (2.5,\above) {$\begin{matrix}  1 \\ 0 \\ 1 \\ \end{matrix}$};
%\node at (  3,\above) {$\begin{matrix}  1 \\ 1 \\ 0 \\ \end{matrix}$};
%\node at (3.5,\above) {$\begin{matrix}  1 \\ 1 \\ 1 \\ \end{matrix}$};

%%%%%%%%%%%%%%%%%%%%%%%%%%%%%%%%%%%%%%%%%%%%%%%%%%%%%%%%%%%%%%%%%%%%%%%%%%%%%%
% KEY AT BOTTOM
%%%%%%%%%%%%%%%%%%%%%%%%%%%%%%%%%%%%%%%%%%%%%%%%%%%%%%%%%%%%%%%%%%%%%%%%%%%%%%
\node at ( -0.5,\beneath) {$\begin{matrix} v_1 \\ v_2 \\ v_3\\ \end{matrix}$};
\node at (  0,\beneath) {$\begin{matrix}  0 \\ 0 \\ 0 \\ \end{matrix}$};
\node at (0.5,\beneath) {$\begin{matrix}  0 \\ 0 \\ 1 \\ \end{matrix}$};
\node at (  1,\beneath) {$\begin{matrix}  0 \\ 1 \\ 0 \\ \end{matrix}$};
\node at (1.5,\beneath) {$\begin{matrix}  0 \\ 1 \\ 1 \\ \end{matrix}$};
\node at (  2,\beneath) {$\begin{matrix}  1 \\ 0 \\ 0 \\ \end{matrix}$};
\node at (2.5,\beneath) {$\begin{matrix}  1 \\ 0 \\ 1 \\ \end{matrix}$};
\node at (  3,\beneath) {$\begin{matrix}  1 \\ 1 \\ 0 \\ \end{matrix}$};
\node at (3.5,\beneath) {$\begin{matrix}  1 \\ 1 \\ 1 \\ \end{matrix}$};

%%%%%%%%%%%%%%%%%%%%%%%%%%%%%%%%%%%%%%%%%%%%%%%%%%%%%%%%%%%%%%%%%%%%%%%%%%%%%%
% LEVEL 0
%%%%%%%%%%%%%%%%%%%%%%%%%%%%%%%%%%%%%%%%%%%%%%%%%%%%%%%%%%%%%%%%%%%%%%%%%%%%%%

\node at (\leftlabel,\levelzero-0.25) {$F$};

% flat
\draw[thick] (-0.25,\levelzero) -- (1.75,\levelzero);

% spike 100
\draw[thick] (1.75,\levelzero) -- (2, \levelzerobottom) node (bottom100) {}; % up
\draw[thick] (bottom100.center) -- (2.25,\levelzero); % up

% flat
\draw[thick] (2.25,\levelzero) -- (2.75,\levelzero);

% spike 110
\draw[thick] (2.75,\levelzero) -- (3, \levelzerobottom) node (bottom110) {}; % up
\draw[thick] (bottom110.center) {} -- (3.25,\levelzero); % up

% flat 
\draw[thick] (3.25,\levelzero) -- (3.75,\levelzero);

%%%%%%%%%%%%%%%%%%%%%%%%%%%%%%%%%%%%%%%%%%%%%%%%%%%%%%%%%%%%%%%%%%%%%%%%%%%%%%
% LEVEL 1
%%%%%%%%%%%%%%%%%%%%%%%%%%%%%%%%%%%%%%%%%%%%%%%%%%%%%%%%%%%%%%%%%%%%%%%%%%%%%%

\node at (\leftlabel+0.35,\levelone-0.15) {\textcolor{blue}{Extender}};

% flat
\draw[thick] (-0.25,\levelone) -- (1.75,\levelone);

% spike 100
\draw[thick] (1.75,\levelone) -- (2, \levelonebottom) node (bottom100) {}; % up
\draw[thick, name path=up100] (bottom100.center) -- (2.25,\levelone); % up

% flat
\draw[thick] (2.25,\levelone) -- (2.75,\levelone);

% spike 110
\draw[thick] (2.75,\levelone) -- (3, \levelonebottom) node (bottom110) {}; % up
\draw[thick] (bottom110.center) {} -- (3.25,\levelone); % up

% flat 
\draw[thick] (3.25,\levelone) -- (3.75,\levelone);

% lengths of these set by hand so as to not increase bounding box
% swap path to draw to see for yourself
\path[name path=from100] (bottom100.center) -- ++(-185:2.6cm);
\path[name path=from110] (bottom110.center) -- ++(-185:2cm);
\path[name path=vertical] (-0.25,\levelone) -- (-0.25,\levelonebottom);

\draw[thick,blue,dashed, name intersections={of= from110 and up100}] (bottom110.center) -- (intersection-1);
\draw[thick,blue,dashed, name intersections={of= from100 and vertical}] (bottom100.center) -- (intersection-1);

%%%%%%%%%%%%%%%%%%%%%%%%%%%%%%%%%%%%%%%%%%%%%%%%%%%%%%%%%%%%%%%%%%%%%%%%%%%%%%
% LEVEL 2
%%%%%%%%%%%%%%%%%%%%%%%%%%%%%%%%%%%%%%%%%%%%%%%%%%%%%%%%%%%%%%%%%%%%%%%%%%%%%%

\node at (\leftlabel+0.35,\leveltwo-0.07) {Output};

% flat
%\draw[thick] (-0.5,\leveltwo) -- (1.75,\leveltwo);

% spike 100
\path[thick] (1.75,\leveltwo) -- (2, \leveltwobottom) node (bottom100) {}; % up
\path[thick, name path=up100] (bottom100.center) -- (2.25,\leveltwo); % up

% flat
%\draw[thick] (2.25,\leveltwo) -- (2.75,\leveltwo);

% spike 110
\path[thick] (2.75,\leveltwo) -- (3, \leveltwobottom) node (bottom110) {}; % up
\draw[thick] (bottom110.center) {} -- (3.25,\leveltwo); % up

% flat 
\draw[thick] (3.25,\leveltwo) -- (3.75,\leveltwo);

% lengths of these set by hand so as to not increase bounding box
% swap path to draw to see for yourself
\path[name path=from100] (bottom100.center) -- ++(-185:2.6cm);
\path[name path=from110] (bottom110.center) -- ++(-185:2cm);
\path[name path=vertical] (-0.25,\leveltwo) -- (-0.25,\leveltwobottom);

\draw[thick, name intersections={of= from110 and up100}] (bottom110.center) -- (intersection-1);
\draw[thick, name intersections={of= from100 and vertical}] (bottom100.center) -- (intersection-1);

%%%%%%%%%%%%%%%%%%%%%%%%%%%%%%%%%%%%%%%%%%
% DRAW THE MIDDLE
%%%%%%%%%%%%%%%%%%%%%%%%%%%%%%%%%%%%%%%%%%
\draw[thick, name intersections={of= from110 and up100}] (bottom100.center) -- (intersection-1);

%%%%%%%%%%%%%%%%%%%%%%%%%%%%%%%%%%%%%%%%%%%%%%%%%%%%%%%%%%%%%%%%%%%%%%%%%%%%%%
\end{tikzpicture} 
}
\caption{\coNP lower bound construction. Troughs encode false assignments.}
\label{fig:conp}
\end{figure}

To obtain \coNP hardness, we use essentially the same technique, but with one
important difference in our encoding of bits. 
In the \NP-hardness result we used the constant $L$ to encode a zero bit, and a
peak that touches the constant $H$ to encode a one bit. To prove \coNP hardness,
we flip that upside down.
\begin{itemize}[leftmargin=0.3cm]
\item If a bit is one during an interval, then the relative value of the state
will remain at $H$ for the entire interval.
\item If a bit is zero during an interval, then this is encoded as a trough,
during which the relative value touches $L$. 
\end{itemize}
We use this encoding, which we call the \emph{reverse encoding} throughout the
\coNP-hardness construction: all of the states of the enumeration gadget use the
reverse encoding, and the formula evaluation is also done in reverse encoding.
We end up with a state whose relative value encodes $F$ in reverse encoding, as
shown in Figure~\ref{fig:conp}.

With the reverse encoding, if $F$ is always true, then the relative value of the
state will be $H$. If there exists an input that makes $F$ false, then this will
be encoded as a trough. We can extend this back to time zero using an extender
state with a carefully chosen holding rate\footnote{As for \NP-hardness, this 
holding rate can be $1/2$.}, though this time the extender state
must be a minimizer state, since we want the extender player to obtain a lower
value by waiting until $F$ is not satisfied.

The end result is that the relative value at time $0$ is $H$ if $F$ is always
true, and it is strictly less than $H$ if there exists an assignment to
variables that makes $F$ false. Again, this construction uses only three holding
rates, so we obtain the following theorem. 
%\begin{theorem}\label{thm:conp_hard}
%For an SPTG, deciding whether $v(s, 0) \ge c$ for a given state $s$ and constant
%$c$ is \coNP-hard, even if the game has only holding rates in $\{0,1/2,1\}$.
%\end{theorem}
\begin{restatable}{theorem}{conphardthm}
\label{thm:conp_hard}
For an SPTG, deciding whether $v(s, 0) \ge c$ for a given state $s$ and constant
$c$ is \coNP-hard, even if the game has only holding rates in $\{0,1/2,1\}$.
\end{restatable}

The proof of this theorem appears in Appendix~\ref{app:np_conp_hardness}. Since
the \NP and \coNP-hardness proofs are very similar, we prove them both at the
same time in the appendix.

\section{\PSPACE lower bound}

\begin{figure*}
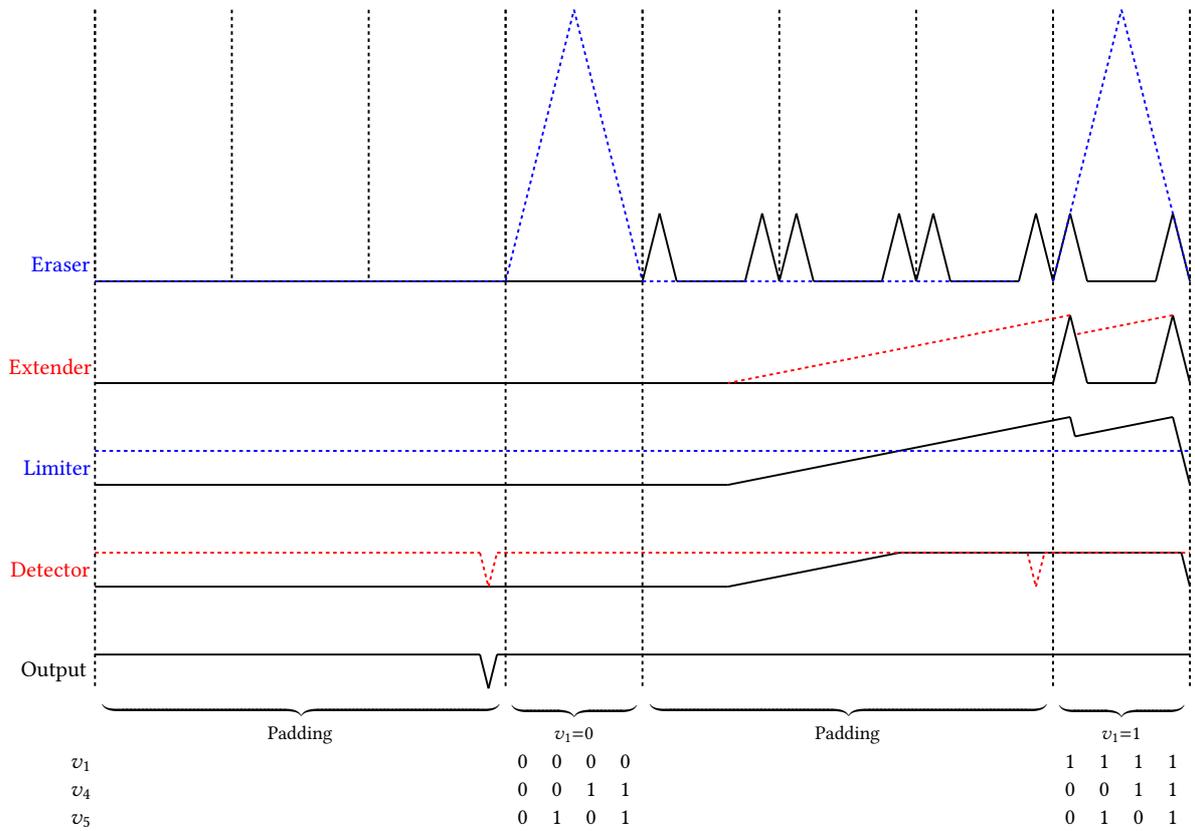

\centering
\scalebox{0.9}{
\input figures/pspace.tex
}
\caption{\PSPACE lower bound construction.}
\label{fig:pspace}
\end{figure*}

We now move on to our main result, and show that computing the value of a
particular state at time zero is \PSPACE-hard. We will reduce directly from
TQBF, which is the problem of deciding whether a quantified Boolean formula is
true. The high level idea is to make use of the techniques from our \NP-hardness
reduction to deal with existential quantifiers, and the techniques from our
\coNP-hardness reduction to deal with universal quantifiers.

As a running example, we will use the formula
\begin{center}
$F = (v_1 \lor v_4 \lor v_5) \land (v_1 \lor \lnot v_4 \lor \lnot v_5)$,
\end{center}
and we will apply the reduction to the TQBF instance 
$$\forall v_1 \; \exists v_4 v_5 \cdot F(v_1, v_4, v_5).$$
The slightly odd choice of variable indices will be explained shortly.

\paragraph{\bf Overview.}

As in previous reductions, we will divide the time period into intervals, and we
will associate each interval with a bit string, and evaluate the formula on each
of those bit strings. However, in this setting we must now deal with both types of
quantifiers. 

Our solution is shown in Figure~\ref{fig:pspace}. We use the
quantifiers to divide the bit strings into blocks, and place padding between the
blocks. For our running example, we have two blocks, which correspond to the
case where $v_1 = 0$ and the case where $v_1 = 1$. So we have split the bit
strings according to the universal quantifier in the formula. We will refer to
the two sub-instances as 
$F'(v_1) := \exists v_4 v_5 \cdot F(v_1, v_4, v_5)$ .

The idea is to evaluate the two blocks separately using the method from the
\NP-hardness reduction. So we will determine whether $F'(v_1)$ holds when $v_1 =
0$, and independently determine whether it holds when $v_1 = 1$. This then
leaves us with the problem of deciding whether $\forall v_1 \cdot F'(v_1)$ is
true, which will be evaluated using methods from the \coNP-hardness reduction.
We do this by turning the output of the two independent evaluations of $F'$
into a reverse encoded input for the \coNP problem.

\paragraph{\bf The padding.}

The padding between the blocks is used to ensure that the two evaluations of
$F'$ are independent. The padding is implemented by inserting extra dummy
variables into the formula. In our running example, we add the extra dummy
variables $v_2$ and $v_3$, but we do not modify the formula itself in any way.
As shown in the first line of Figure~\ref{fig:pspace}, this leads to each block
being repeated four times, since we enumerate all four possible settings for
$v_2$ and $v_3$, but none of these change the output of the formula. 

The first step is to take the minimum of this relative value function with a
state that we call the \emph{eraser}, whose relative value function is shown in
blue in the first line of Figure~\ref{fig:pspace}. This value function peaks
during the block\footnote{We do so by encoding the formula $(v_2\wedge v_3)$
using our previous constructions for formulas but only over these two variables,
which results in the very high peak.} that we would like to keep, but stays at
the value $L$ during the blocks that we would like to erase. So by taking the
minimum, we keep the right-most copy, and erase the other three, which gives us
the padding between the blocks.

Recall from the \NP-hardness reduction that the extender
state is used to detect whether the relative value has peaked during a block.
Furthermore, the relative value of the extender decays over time, and that the
rate at which this happens is controlled by the holding rate of the
extender. 

In the \PSPACE-hardness reduction, we choose the decay rate so that the value
will always decay back to $L$ during the padding before the next block starts.
This can be seen in the extender state in the second line of
Figure~\ref{fig:pspace}. In the right-hand block, there are two assignments that
make the formula true, and this information is carried to the left-hand edge of
the block by the extender. The padding provides enough space for the extender to
decrease back to $L$ before the left-hand block begins.

\paragraph{\bf Changing the encoding.}
So far we have independently evaluated $F'(v_1)$ for both possible settings of
$v_1$, and this is encoded in the second value function of
Figure~\ref{fig:pspace}. The rest of the steps in that  figure show how we then
turn this into a reverse encoding of $\forall v_1 \cdot F'(v_1)$. 

The overall goal is to detect whether the extender is above $L$ at the left-hand
boundary of each block. In fact, we choose the decay rate of the extender to be
slow enough to guarantee that if there was a peak during the block the
value of the extender is above $(H + L)/2$ at the left-hand edge of the block. 

The first step is to take the minimum of the relative value function with a
\emph{limiter} state, shown in blue in the third line of
Figure~\ref{fig:pspace}, whose relative value is constant at $(H + L)/2$. This
effectively chops off the top half of the function. We then construct a state,
known as the \emph{detector}, shown in red in the fourth line of
Figure~\ref{fig:pspace}. This state has a relative value function
that remains
at $(H + L)/2$ throughout, except at the left-hand edge of each block, where
there is a trough that touches $L$.
We do this  by encoding the formula $(\neg v_2\vee v_3\vee \neg v_4\vee \neg
v_5)$ in reverse encoding.

We take the maximum of the value function with the detector. This does the
following.
\begin{itemize} [leftmargin=0.3cm]
\item If there was a peak during the block, the value of the extender will be
above $(H + L)/2$, and so the trough in the detector will be eliminated. The limiter ensures
that the relative value does not exceed 
$(H + L)/2$ in this case.
\item If there was no peak during the block, the value of the extender will be
$L$, and so the trough in the detector will not be eliminated.
\end{itemize}
The end result is that we have a trough in the final value function if and only
if $F'(v_1)$ was false for the corresponding block.

Observe that this is a valid reverse encoding of the problem $\forall v_1 \cdot
F'(v_1)$, with the only change being that the relative function ranges between
$L$ and $(H+L)/2$ rather than $L$ and $H$. So we can apply the techniques from
the \coNP-hardness reduction to determine whether 
$\forall v_1 \cdot F'(v_1)$ is true.

\paragraph{\bf \PSPACE-hardness.}

So far, we have seen how to deal with alternations of the form $\forall x
\exists y$, but the same techniques can also deal with alternations of the
form $\exists y \forall x$. The only difference is that we must turn a
reverse encoded output into the normal encoding, which can again be done with
appropriately constructed limiter and detector states.

The full \PSPACE-hardness result applies the two techniques inductively.
Every alternation of quantifiers in the formula is handled by turning one
encoding into the other, ready to be evaluated by the next level of quantifiers.
The full details can be found in Appendix~\ref{app:pspace}, where we prove the
following result.

%\begin{theorem}\label{thm:pspace}
%For an SPTG, deciding whether $v(s, 0) \ge c$ for a given state $s$ and constant
%$c$ is \PSPACE-hard.
%\end{theorem}
\begin{restatable}{theorem}{pspacethm}
\label{thm:pspace}
For an SPTG, deciding whether $v(s, 0) \ge c$ for a given state $s$ and constant
$c$ is \PSPACE-hard.
\end{restatable}

It is also worth noting that if the formula only has $k$ alternations, then the
resulting game uses $k+2$ distinct holding rates. The holding rates $0$ and $1$
are already used by the exponential lower bound game. Each level of alternation
uses an extender state with a distinct holding rate, which accounts for the
other $k$ holding rates. Hence, we also get the following result.

%\begin{theorem}\label{thm:poly_hier} 
%For an SPTG with $k+2$ distinct holding rates, deciding whether $v(s, 0) \ge c$
%for a given state $s$ and constant $c$ is hard for the $k$-th level of the
%polynomial-time hierarchy.
%\end{theorem}
\begin{restatable}{theorem}{polyhierthm}
\label{thm:poly_hier} 
For an SPTG with $k+2$ distinct holding rates, deciding whether $v(s, 0) \ge c$
for a given state $s$ and constant $c$ is hard for the $k$-th level of the
polynomial-time hierarchy.
\end{restatable}

\subsection{Other decision problems}
All of our results so far have shown hardness of deciding whether $v(s, 0) \ge
c$ for some state $s$, and some constant $c$. In this section, we point out that
our construction can also prove hardness for other, related, decision problems. 

As in the \NP- and \coNP-hardness section, we can let the outer-most extender
state produce a horizontal line, rather than a decaying one. 
This ensures that we can pick two constants $H'$ and $L'$ such that
$\val(v,0)=H'$
if the formula is true and $\val(v,0)=L'$ if the formula is false.
Thus, all our hardness proofs for each of \NP-, \coNP- and \PSPACE-hardness, and
hardness for the $k$-th level of the polynomial time hierarchy, apply to the
following promise problem.

\paragraph{\bf PromiseSPTG:} Given an SPTG, a state $v$ and two numbers
$c>c'$, with the promise that $\val(v,0)\in \{c,c'\}$, is ${\val(v,0)=c}$?

\noindent This problem can be reduced in polynomial time to each of the following
problems.
\begin{enumerate} [leftmargin=0.4cm]
\item DecisionSPTG: Given an SPTG $G$, a state $v$ and a value $c$, is $\val(v,0)\geq c$?
\item EqualDecisionSPTG: Given an SPTG $G$, a state $v$ and a value $c$, is $\val(v,0)= c$?
\item $\epsilon$-StrategySPTG: Given an SPTG $G$, a state $v$, an $\epsilon> 0$ and an action $a$ is there an $\epsilon$-optimal strategy that uses~$a$ at time $0$?
\item StrategySPTG: Given an SPTG $G$, a state $v$ and an action $a$ is there an optimal strategy that uses $a$ at time~$0$?
\item AllOptimalStrategiesSPTG: Given an SPTG $G$, a state $v$ and an action $a$
do all optimal strategies use $a$ at time~$0$?
\end{enumerate}
The reduction is trivial for (1) and (2), since we have just removed the
promise. 

For (3), (4) and (5), fix some $\epsilon<\frac{c-c'}{2}$. We add
another minimizer state $v'$ to the game with holding rate $M+1$, where $M$ is
the largest holding rate in the rest of the game. The state $v'$ has an edge to
$v$ and an edge to a goal state. The edge to $v$ has cost $0$ and the edge to
the goal state has cost~$\frac{c+c'}{2}$.

It is clear that $\val(v',0)=c'$ if and only if $\val(v,0)=c'$. Also, if
$\val(v,0)=c'$, then no $\epsilon$-optimal strategy can use the edge to $v$ at
time $0$, so no optimal strategy can do this either. Similarly, if
$\val(v,0)=c$, then no $\epsilon$-optimal strategy can use the edge to the goal
state. This proves hardness for (3) and (4). 
Also, since the holding rate is larger than $M$ in the above
construction, the minimizer will not wait in $v'$ under an optimal strategy and
therefore he must use an edge immediately, which proves hardness for (5).

That said, it is $\epsilon$-optimal, for any
$\epsilon>0$, to wait for a duration of $\frac{\epsilon}{M}$ in $v'$ and then
make the optimal choice in either cases, when starting in $v'$ at time 0.
This explains why the $\epsilon$-optimal variant of AllOptimalStrategiesSPTG
does not appear in our list.

Note that parametrising the problems with time $t$, instead of always using
time~$0$, trivially makes the questions even harder. Also, using techniques
similar to what we use for shifting in our construction allow us to show
hardness for any of these problems for a given time $t\in(0,1)$. Finally, as
shown by~\cite{HIM13}, finding $\val(v,1)$ and the optimal and
$\epsilon$-optimal choice at time~1 can be solved in time $O(m+n\log n)$ and
is thus in \P.

\section{Properties of our hard instances}
\label{sec:graph_prop}

The instances that we have constructed actually lie in a very restricted class
of graphs, which we describe in this section. 

\paragraph{\bf The exponential-many event point games.}
In Section~\ref{sec:exponential} the family of games that we constructed are all
DAGS with degree four, as seen in Figure~\ref{fig:basiclb}. In
Appendix~\ref{app:graph_prop}, we show that by slightly modifying the graph,
this can actually be reduced to a {\bf DAG with degree three}.
Furthermore, there are only {\bf two distinct holding rates} namely the ones in $\{0,1\}$.

%In essence, our family $F$ that has an exponential number of event points is a
%{\bf DAG with degree four}, see , (we can easily lower it to {\bf degree
%three}, by, for each state $v$ of our construction, adding a new state $v'$,
%which is a minimizer state with rate $1$ with a single edge to $v$ and then
%changing all edges that goes to $v$ to instead go to $v'$. All original states
%have the same value function as before and all newly added states have the same
%value function as their corresponding state. Also, doing so does not change any
%of the other graph properties). 

%\footnote{i.e. that it can be drawn in the plane without edges
%crossing} 

The game also has a {\bf planar graph}. This can be seen by redrawing
Figure~\ref{fig:basiclb} in the following way. The crossing of edges in the
middle of each level can be eliminated by taking each edge
$(v_{r}^i,v_{\ell}^{i-1})$ and making it 
``wrap-around'' under the structure by passing $v_{r}^0$ on the right before going to
the left side and moving up.

While proving an upper bound on the number of event points in a class of games
similar, but more general than ours, the authors of \cite{R11,BLMR06} use a
technique based on adding more and more \emph{urgent} states to the game. A
state is urgent if the owner is not allowed to wait in it. In our construction
with exponentially-many event points, the minimizer would not want to wait in a
state with rate $1$, and the maximizer would not want to wait in a state with
rate 0, because in both cases this is the worst possible rate for them. So the
optimal strategies only wait in the state $v_{r}^0$. Therefore, making any
number of states, besides $v_{r}^0$, urgent does not change the value functions.
Hence, our results show that, while games with no non-urgent states are easy
(because they can be solved as a priced game) games with a \emph{single}
non-urgent state still may have an exponential number of event points.

In Appendix~\ref{app:graph_prop} we give a more in depth argument for this and
also argue that each member of the family have {\bf pathwidth,
treewidth, cliquewidth and rankwidth three}.

\paragraph{\bf The \PSPACE-hard games.}

The \PSPACE-hard games add several extra gadgets to the exponentially-many 
event point games.
These gadgets essentially form a directed tree structure, whose leafs have
outgoing edges to a unique copy of one of our exponential lower bound games. Hence, the
games continue to be {\bf DAGs} and {\bf planar} (because no edge goes ``over'' 
the top states in our exponentially-many event point games), and the gadgets can also be
constructed so that the games continue to have {\bf degree three}. 
In Appendix~\ref{app:graph_prop} we give a more in depth argument for this and
also argue that each member of the family have {\bf treewidth,
cliquewidth and rankwidth} three. We lose bounded pathwidth as a property,
which is caused by the large tree of states that we add to construct our gadgets.

For the \NP, \coNP, and polynomial time hierarchy results, we show in
Appendix~\ref{app:graph_prop} that a variant of our constructions (that is not
planar and where the treewidth, cliquewidth and rankwidth is 4 instead of 3) has
the property that for \NP and \coNP hard instances there are only 2 states
that cannot be made urgent.
Each alternation adds one extra state that cannot be made urgent.
Hence, it is \NP-hard to solve games with 2 non-urgent states and hard for the
$k$-th level of the polynomial time hierachy to solve games with $k+1$
non-urgent states. 

\section{Upper bounds for undirected graphs, trees and DAGs}
\label{sec:upper}

In Section~\ref{sec:exponential} and Section~\ref{sec:graph_prop}, we showed
that there are an exponential number of event points for SPTGs belonging to even
very restrictive graph classes. In this section we show that there are 
some classes of games in which there is at most a polynomial number of event points.
Specifically, this holds for undirected graphs and trees. 
It then follows by~\cite{HIM13} that the event point
iteration algorithm algorithm runs in polynomial-time
for these problems.

Secondly, we show that SPTGs on DAGs are in \PSPACE. The result extends to OCPTGs because of a reduction by~\cite{HIM13}. Our main result implies
that they are \PSPACE-hard and thus, this shows that they are \PSPACE-complete.

\paragraph{\bf Undirected graphs.}
The trick is to consider that whenever play goes to a maximizer state $v$ at
some time $t$ from some other state $v'$, the maximizer can choose to send the
play immediately back to state $v'$. Because our strategies are time-positional,
if the maximizer follows this strategy, and the play then ever goes to a maximizer state $v$ from some
other state $v'$, the play will continue going back and forth between $v$ and
$v'$ forever, and therefore never reach a goal state. The outcome is therefore
$\infty$, which is the best possible for the maximizer, and so we can assume
that he will adopt this strategy.

As shown in~\cite{HIM13}, if $\val(v,t')=\infty$ for some $t'$, then
$\val(v,t)=\infty$ for all $t$ and $\val(v,1)$ can be found in $O(m+n\log n)$.
In the remaining states, we can assume that maximizer states cannot be entered.
This allows us to solve the minimizer and goal states as a sub-game first (which
can be done in polynomial time since it is a priced timed
automata~\cite{LMS04,HIM13}). The remaining maximizer states are also easy to solve in
polynomial time once this has been done. Full details of the 
argument can be found in Appendix~\ref{app:pos}.

\paragraph{\bf Trees.}
The argument for trees is also fairly straightforward, in that the following
lemma (see Appendix~\ref{app:pos}) can easily be shown, using structural
induction and how value functions are computed by the value iteration algorithm.

We will say that a line segment $L$ is covered by a line or line segment $L'$ if $L\subseteq L'$ and also extend the notion to  sets, i.e. a set $S$ is covered by a set $S'$ if each element $L\in S$ is covered by some element $L'\in S'$ (which may depend on $L$).
%\begin{lemma}\label{lem:trees}
%Consider a state $s$ which is the root of a tree with $k$ leaves, for some number $k$.
%Then, let $L_s$ be the line segments of $\val(s,t)$.
%There exists a set $L_k$ of $k$ lines that covers $L_s$.
%\end{lemma}
\begin{restatable}{lemma}{treeslem}
\label{lem:trees}
Consider a state $s$ which is the root of a tree with $k$ leaves, for some number $k$.
Then, let $L_s$ be the line segments of $\val(s,t)$.
There exists a set $L_k$ of $k$ lines that covers $L_s$.
\end{restatable}

Because there can be at most $\frac{k(k-1)}{2}$ intersections of $k$ lines, that
is also a bound on the number of line segments of $\val(v,t)$. This, in turn,
means that there are at most $O(n^3)=O(nk^2)$ many line segments for $\val(v,t)$
over all $n$ states of the graph. Because an event point is the time coordinate
of an end point of some line segment of $\val(v,t)$ for some state $v$, we therefore have at most $O(n^3)$ event points.

\paragraph{\bf DAGs.}
To show that DAGs are in \PSPACE, we will first argue that the denominator of
each event point $t^*$ and number $\val(v,t^*)$ for all $v$ can be expressed in polynomial space.
For a natural number $c$, we say that a fraction $x$ is {\em $c$-expressible} if the denominator $d$ of $x$ is such that $d\cdot k=c$ for some natural number $k$. 
In Appendix~\ref{app:pos}, we show the following lemma.

%\begin{lemma}\label{lem:dags}
%Consider a SPTG on a DAG of depth $h$ with integer holding rates. Let $R=\Pi_{v_1,v_2\in V\mid r(v_1)\neq r(v_2)}\abs{r(v_1)-r(v_2)}$. 
 %Let $v$ be some state at depth $h_v$ and $(x,y)$ some end point of a line segment of $\val(v,t)$. 
%If $y=\infty$, then $\val(v,t)=\infty$ and otherwise,
 %if $y\neq \infty$, the numbers $x$ and $y$ are $R^{h-h_v}$-expressible.
%\end{lemma}
\begin{restatable}{lemma}{dagslem}
\label{lem:dags}
Consider a SPTG on a DAG of depth $h$ with integer holding rates. Let $R=\Pi_{v_1,v_2\in V\mid r(v_1)\neq r(v_2)}\abs{r(v_1)-r(v_2)}$. 
 Let $v$ be some state at depth $h_v$ and $(x,y)$ some end point of a line segment of $\val(v,t)$. 
If $y=\infty$, then $\val(v,t)=\infty$ and otherwise,
 if $y\neq \infty$, the numbers $x$ and $y$ are $R^{h-h_v}$-expressible.
\end{restatable}

We can find the set of states for which $\val(v,t)=\infty$ in time $O(m+n\log
n)$ by using techniques from~\cite{HIM13}. Specifically, that paper shows that
if $\val(v,t')=\infty$ for some $t'$ then $\val(v,t)=\infty$ for all $t$, and
that paper also shows how to find $\val(v,1)$ in time $O(m+n\log n)$ for all
$v$.

For the remaining states, note that $R^{h-1}$ can be described in polynomial
space, since it is a product of $\leq hn^2$ numbers, each of which are bounded
by the largest holding rate. In turn, we can also bound the numerators as using at most
polynomial space and thus all the numbers.

This, in turn, means that a variant of the event point iteration algorithm given
in~\cite{HIM13} (that does not store the end points of line segments of
$\val(v,t)$, which is only used for the output) runs in polynomial space (see
Appendix~\ref{app:known_algo} for pseudo-code for the event point iteration
algorithm), because it then stores only $t^*$ and $\val(v,t^*)$ for all $v$ at
any one point for some event point $t^*$. That can then find $\val(v,t')$ for
some given $v,t'$ by finding the value $\val(v,t^*)$ for the smallest event
point $t^*>t$ and how $\val(v,t)$ behaves between $t^*$ and the next smaller
event point (which is how the algorithm iterates over the event points). Thus,
it can solve the decision question we are interested in.
We give more details of this argument in Appendix~\ref{app:pos}.

\paragraph{Games with holding rates $\{0,1\}$.}
In~\cite{BGNMMT14}, it was previously claimed (fixed in the latest arXiv
version 6, \href{https://arxiv.org/abs/1404.5894v6}{arXiv:1404.5894v6}) that an
SPTG with holding rates $\{0,1\}$ and integer costs can be solved in polynomial
time because, it was claimed, such games would have only polynomially-many event
points. Our results, however, show that this claim is incorrect: We show in
Appendix~\ref{app:ashutosh} how to convert our examples with exponentially-many
event points and holding rates $\{0,1\}$ to have integer costs. 

\section{Open questions}
\label{sec:open_questions}
While our results show that one-clock priced timed games and many special cases
are \PSPACE-hard, there are still a number of open questions.

The biggest open question for priced timed games is likely the complexity of
two-clock priced timed games. That said, a number of other models related to
priced timed games have been considered and there is often a jump in complexity
when going from one clock to two or more clocks in those models, as we mentioned
in the introduction. Also, many questions related to three or more clocks for
priced timed games are undecidable~\cite{BBM06,BBR05,BJM15}. This suggests that
similar questions for the case of two clocks are also undecidable.

Besides that, we show that the complexity of priced timed games is \PSPACE-hard.
Previous work have shown them to be solvable in exponential
time~\cite{R11,HIM13}, which does leave a gap. A possible way to resolve the
question is to show a conjecture by~\cite{BLMR06}. If, as conjectured
by~\cite{BLMR06}, the number of iterations of the value iteration algorithm is
polynomial, the problem is \PSPACE-complete, since DAGs are in \PSPACE, as we
show, and the value iteration algorithm in essence turns the game into a DAG
with states polynomial in the number of iterations and the number of states of
the game.

Let $\ell$ be the number of event points. We show $\ell\geq 2^{n/2}/2$. Previous
work~\cite{HIM13} has shown that $\ell\leq 12^n$ for SPTGs. This means that
$\ell=2^{\Theta(n)}$, but this is quite a wide gap, and one could work on making
it smaller.

We have shown that priced timed games on DAGs with one clock are
\PSPACE-complete, but the best result for DAGS with more clocks~\cite{ABM04} is
exponential. For DAGs the results for more clocks seems similar to the one clock
case though: The value iteration algorithm runs in exponential time
(see~\cite{ABM04} for the upper bound on more clocks, and we show the lower
bound for one clock). There is an exponential number of areas (called event
points for one clock) where the strategy should change (see~\cite{ABM04} for
more clocks, we have the lower bound for one clock and the upper bound for one
clock is in~\cite{HIM13,R11}). Does our \PSPACE upper bound generalise to more
clocks?

While we resolve several special cases of one-clock priced timed games, a number
are still open: 

\medskip

\begin{itemize}[leftmargin=0.3cm]
\itemsep2mm

\item {\em Constant pathwidth.} We show that each member of our family that has
	exponentially many event points has pathwidth 3, but no
	computational-complexity hardness is shown and it is plausible that they are
	easier than the general case.

\item {\em Pseudo-polynomial time algorithm for costs.} Our 
	constructions use costs that double as we double the number of event points.
	To avoid the lower bounds, one could consider pseudo-polynomial time
	algorithms.

\item {\em A player with few states.} Our \PSPACE-hard construction has a nearly
	equal number minimizer and maximizer states. On the other hand, the for
	automata (i.e., when only one player has states) the corresponding problem
	is in \NLOGSPACE~\cite{LMS04}. Can one design fast algorithms for the case
	where one player only has a few states? Here, few could mean either constant
	or one could do a parametrized analysis.

\item {\em Very limited graph width.} We show hardness for games with
	treewidth, cliquewidth and rankwidth three, but the cases of lower
	treewidths, cliquewidths and rankwidths are still open (apart from trees,
	which we have shown in Section~\ref{sec:upper} can be solved in polynomial time).

\end{itemize}

\clearpage

%% Bibliography
\bibliographystyle{plain}
\bibliography{bibliography}

\clearpage

%% Appendix
\appendix

\section{Two known algorithms for OCPTGs}
\label{app:known_algo}
In this section we describe two known algorithms that we 
will use in our analysis.
There first algorithm iterates over how many steps
one is allowed to take before the game is over.
The second algorithm iterates over the event points of the game.

\paragraph{The value iteration algorithm.}
A variant of the algorithm that iterates over how many steps one is allowed to take before the game is over is used for many classes of games and is typically called the {\em value iteration algorithm}.
The algorithm was defined and shown correct independently by~\cite{ABM04,BCFL05}.
The algorithm, given a game $G$, is based on defining the notion of a finite-horizon game $G^k$, where $k$ is some natural number.
In $G^k$, the outcome is $\infty$ if more than $k$ steps are taken.
This definition allows one to find $\val(v,t,G^k)$ easily from the value functions for its successors in $G^{k-1}$ (because, when entering the successors, one less step is left). 
For a piecewise linear function $f(t)$, where $t\in[0,1]$ let $L(f(t))$, be the set of end points of line segments of $f(t)$.
The function $\up$ (resp. $\low$) is the upper (resp. lower) envelope of a set of functions, i.e. basically max (resp. min), but for functions. Formally, let $f_1,\dots,f_{\ell}:[0,1]\rightarrow \R$, for some number $\ell$, then \[
\up(f_1,\dots,f_{\ell})(x)=\max(f_1(x),\dots,f_{\ell}(x))
\]
and
\[
\low(f_1,\dots,f_{\ell})(x)=\min(f_1(x),\dots,f_{\ell}(x)) \enspace .
\]
For a fixed $i$ and $x$, the number $f_i(x)$ is omitted from the max or min, if it is undefined and in turn the functions $\up,\low$ are undefined at $x$ if all of $f_i(x)$ are.
Also, given two points $(x,y),(x',y')$, we let $(x,y)-(x',y')$ be the line between them.
We give pseudo-code for the algorithm in Algorithm~\ref{alg:val_it}.
As shown by~\cite{BLMR06}, \[
\lim_{k\rightarrow \infty} \val(v,t,G^k)=\val(v,t,G).
\]

\begin{algorithm}
\SetAlgoLined
\KwResult{$\val(v,t,G^k)$ for all $v,t$}
 \For{$v\in V$} {
 	\eIf{$v$ is a goal state} {
 		$\val(v,t,G^0)=0$
 	} {
 		$\val(v,t,G^0)=\infty$
 	}
 }
 \For{$(k'\leftarrow 1;k'\leq k;k'\leftarrow k'+1)$}{
  \For{$v\in V$} {
 	\eIf{$v$ is a goal state} {
 		$\val(v,t,G^{k'})=0$\;
 	} {
 		$S\leftarrow \emptyset$\;
		\For{$(v,u)\in E$} {
			$S\leftarrow S\cup\{\val(u,t,G^{k'})+c((v,u))\}$\;
			\For{$(x,y)\in L(\val(u,t,G^{k'})+c((v,u)))$}{
				$S\leftarrow S \cup \{(0,y+r(v)x)-(x,y)\}$\;
			}
		}
		
		\eIf{$v\in V_1$} {
			$\val(v,t,G^{k'})=\low(S)$\;
		} {
			$\val(v,t,G^{k'})=\up(S)$\;
		}			
 	}
  }
 }
 \caption{Value iteration algorithm\label{alg:val_it}}
\end{algorithm}

\paragraph{The event point iteration algorithm.}
The second algorithm iterates over the event points. 
The algorithm was given and shown correct by~\cite{HIM13}.
In particular, given an event point $t'$ and $\val(v,t')$ for all $v$, it finds the largest event point $t''< t'$ and $\val(v,t'')$ for all $v$. This is done using that if one starts waiting at time $t''$, then one waits until time at least $t'$.
Also, $\val(v,1)$ is easy to find for all $v$, since one cannot wait any more and the game turns in to a priced game.
As shown by~\cite{KBBEGRZ08}, such games can be solved in $O(m+n\log n)$ time, using an algorithm similar to Dijkstra's shortest path algorithm. The edge costs in the following priced games are lexicographic ordered pairs (but we omit the second component if it is 0).
To give pseudo-code for the algorithm, for an SPTG $G$ and a function $f:V\rightarrow \R_+$, let $\PG(G)$ be the priced game with the same states and edges as $G$ and let $\PG(G,f)$ be the extension of $\PG(G)$, that, for all $v$ has an edge $(v,g)$ of cost $f(v)$ where $g$ is a goal state.
To define a piecewise-linear function, one just needs to define the set of end points of line segments. Therefore, to define $\val(v,t)$, we will just define $L(\val(v,t))$, i.e. the end points of line segments of $\val(v,t)$.
We give pseudo-code for the algorithm in Algorithm~\ref{alg:ev_it} (the number $t^*$ becomes in turn each of the event points, starting from the last at 1, and ending with the first at 0).

\begin{algorithm}
\SetAlgoLined
\KwResult{$\val(v,t)$ for all $v,t$}
Solve $\PG(G)$\; $t^*\leftarrow 1$\;	
 \For{$v\in V$} {
 	$L(\val(v,t))\leftarrow \{(1,\val(v,\PG(G)))\}$\;
 	$f(v)\leftarrow (\val(v,\PG(G)),r(v))$\;	
 }

 \While{$t^*\geq 0$} {
	Solve $\PG(G,f)$\;
	$d\leftarrow t^*$\;
	\For{$(v,u)\in E$} {
		$(x,y)\leftarrow \val(v,\PG(G,f))$\;		
		$(x',y')\leftarrow \val(u,\PG(G,f))+(c((v,u)),0)$\;				
		\If{$v\in V_1$ and $y'<y$ and $x'>x$ and $d>(x'-x)/(y-y')$}{
			$d\leftarrow (x'-x)/(y-y')$\;
		}
		\If{$v\in V_2$ and $y'>y$ and $x'<x$ and $d>(x-x')/(y'-y))$}{		
			$d\leftarrow (x-x')/(y'-y)$\;
		} 
	}
	\For{$v\in V$} {
		$(x,y)\leftarrow \val(v,\PG(G,f))$\;	
		$L(\val(v,t))\leftarrow 	L(\val(v,t))\cup \{(t^*-d,x+yd)\}$\;
		$f'(v)\leftarrow (x+yd,r(v))$\;	
	}
	$f\leftarrow f'$\;
	$t^*\leftarrow t^*-d$\;
 } 
 \caption{Event point iteration algorithm\label{alg:ev_it}}
\end{algorithm}

\section{Exponentially-many event points}
\label{app:exp_event_points}

In this section we provide the technical details and proofs for our basic
construction with exponentially-many event points.

Consider the following graph $G$.
The graph is divided into levels, with two states per level. We will divide the states into left and right states. For all $i$, the left state of level $i$ is $v_{\ell}^i$ and the right state is $v_{r}^i$. 
On level~0, the left state, $v_{\ell}^0$, is the goal state and the right state, $v_r^i$, is a maximizer state with holding rate~1.
The state $v_r^i$ has an edge to $v_{\ell}^0$ of cost 0. 
For all $i\geq 1$, at level $i$, the left state, $v_{\ell}^i$ is a minimizer state of holding rate~1 and the right state, $v_{r}^i$ is a maximizer state of holding rate~0. Each node $v\in \{v_{\ell}^i,v_{r}^i\}$ has an edge to $v_{\ell}^{i-1}$ of cost $2^{-i}$ and an edge to $v_{\ell}^{i-1}$ of cost $0$.

We will in this section argue that $G$ has $2^n$ many event points.

\begin{lemma}
\label{lem:val_func_exp}
The value function for each node at level $i$ consists of $2^i$ line segments, each of duration (in time) $2^{-i}$, with slope alternating between $0$ and $-1$. The first line segment of $\val(v_{\ell}^i,t)$ has slope $0$ and starts at value $1-2^{-i}$ and the first line segment of $\val(v_{r}^i,t)$ has slope $-1$ and starts at value 1. 
Furthermore, $\val(v_{\ell}^i,t)+2^{-i-1}$ intersects $2^{i+1}$ times with $\val(v_{r}^{i},t)$ on a line $L$ with slope $-1/2$, starting at $1-2^{-i-1}$. More precisely, the intersections are at time $t=2^{-i-2}+k\cdot 2^{-i-1}$, for $k\in \{0,\dots,2^{i+1}-1\}$. Finally, at time $k\cdot 2^{-i-1}$, for $k\in \{0,\dots,2^{i+1}-1\}$, we have that $|\val(v_{\ell}^i,t)+2^{-i-1}-\val(v_{r}^{i},t)|=2^{-i-1}$
\end{lemma}

\begin{proof}
The proof will be by induction in the level. 

\paragraph{The base case, level~0.}
We see that the value function for $v_{\ell}^0$ (being a goal state) is $0$, satisfying the statement. 
The optimal strategy in state $v_{r}^0$ is to wait until time 1 and then move to goal. At time $t$, we have $t-1$ duration left, costing $t-1$. Thus, the value falls from $1$ at time 0 to $0$ at time or, equivalently, $\val(v_{r}^0,t)=1-t$. Thus, satisfying the lemma statement.

\paragraph{The induction case, level~$i$.}
By induction, we have the lemma statement for level~$i$ and needs to show it for level $i+1$, for some $i\geq 0$. 
Let $f_{\ell}(t)=\val(v_{\ell}^{i},t)+2^{-i-1}$ and let $f_r(t)=\val(v_{r}^{i},t)$ for all $t$.

The value function \[
\val(v_{\ell}^{i})=\low(\val(v_{\ell}^{i})+2^{-i-1},\val(v_{r}^{i}))=\low(f_{\ell}(t),f_r(t))
\]
and the value function \[
\val(v_{r}^{i})=\up(\val(v_{\ell}^{i})+2^{-i-1},\val(v_{r}^{i}))=\up(f_{\ell}(t),f_r(t)) \enspace .
\]

Because of (1)~the alternations of slopes $0$ and $-1$, (2)~them starting with different slopes; and (3)~each line segment having equal duration (of $2^{-i}$), we see that at all times (except for the ends of the line segments at time $k2^{-i}$ for some integer $k$), one of the value functions $\val(v_{\ell}^{i})$ and $\val(v_{r}^{i})$ have slope $-1$ and the other has slope $0$.

We will use the following claim to show the lemma.

\begin{claim}
For each $k\in\{1,\dots,2^{i}\}$, we have the following:
\begin{enumerate}
\item The values of the functions
$f_{\ell}(t_k)$ and $f_r(t_k)$ at time $t_k:=k\cdot 2^{-i}-2^{-i-1}$ are equal (and are otherwise different between time $(k-1)2^{-i}$ and $k2^{-i-1}$), i.e. the functions $k$-th line segment intersect in the middle.  
\item Also, at time $t_k':=k\cdot 2^{-i}$, we have that $\abs{f_r(t_k')-f_{\ell}(t_k')}=2^{-i-1}$, i.e the functions differ by $2^{-i-1}$, at the end of the line segments, when the slopes alternate (it is also the case at time 0).
\end{enumerate}

\end{claim}

\begin{proof}
We will show the claim by induction in $k$. 
First, $k=1$.
By induction in $i$, we have that $f_{\ell}(t)$ starts at $1-2^{-i}+2^{-i-1}=1-2^{-i-1}$ and in the first line segment, of duration $2^{-i}$, it has slope $0$. On the other hand, $f_{r}(t)$ starts at $1$ and for the first line segment, of duration $2^{-i}$, it has slope $-1$. Thus, at time $2^{-i-1}(=t_1)$ they intersect, as wanted. Also, $f_{\ell}(t_1')=1-2^{-i-1}$ and $f_r(t_1')=1-2^{-i}$, for $t_1'=2^{-i}$. Thus, $\abs{f_r(t_1')-f_{\ell}(t_1')}=2^{-i-1}$.

Next, consider $k\geq 2$.
By induction in $k$, we have that $\abs{f_r(t_{k-1}')-f_{\ell}(t_{k-1}')}=2^{-i-1}$ and that the function 
$f_{\ell}(t)$ intersected with $f_{r}(t)$ at time $t_{k-1}=t_{k-1}'-2^{-i-1}$. At time $t_{k-1}'$ (by induction in $i$), the slopes of $f_{\ell}(t)$ and $f_r(t)$ alternate. Because the two functions intersected in the middle of the last line segment, the function with least (resp. highest) value must have slope 0 (resp. $-1$) in the next line segment. They therefore intersect after a duration of $2^{-i-1}$ into the line segment, i.e. at time $t_{k-1}'+2^{-i-1}(=t_k)$ and differ by $2^{-i-1}$ after a duration $2^{-i}$ (which is also the duration of the line segment), i.e at time $t_{k-1}'+2^{-i}(=t_k')$, as wanted. 
\end{proof}
Because of the use of $\low$, resp. $\up$ in the above definition of the value function for  $v_{\ell}^{i}$, resp. $v_{r}^{i}$, the first part of the lemma follows from the first part of the claim. The second part of the lemma (about the first line segment), follows from that (1)~$f_{\ell}(t)\leq f_{r}(t)$ until time $2^{-i-1}$, (2)~that the first line segment of $f_{\ell}(t)$ has slope 0 and starts at value $1-2^{-i}+2^{-i-1}=1-2^{-i-1}$; and (3)~that the first line segment of $f_{r}(t)$ has slope $-1$ and starts at value 1. All three statements comes from induction in $i$, but (1) also uses that the functions first intersect at time $2^{-i-1}$, which comes from the claim.
The third and fourth part of the lemma also comes from the claim.
The lemma thus follows from the claim.\end{proof}
 
The lemma implies the following theorem as a corollary:
 
%\begin{reptheorem}{thm:exp_event}
%There is a family of simple priced time games that have exponentially many event
%points.
%\end{reptheorem}
\expeventthm*

\subsection{Inapproximability with few change points}
\label{app:no_approx}

We will in this section argue that a strategy for the minimizer (resp. maximizer)
such that if there exists a duration $[x\cdot 2^{-k+1},(x+1)\cdot 2^{-k+1})$,
for some integer $k,x$ in which there are no strategy change points, then the
strategy is not $\epsilon$-optimal for $\epsilon<2^{-k}$ when starting in
$v_{\ell}^k$ (resp. $v_{r}^k$). In particular, such a duration exists if there
are $<2^{k-1}$ many strategy change points in the strategy or if there exists
any duration of length $2^{-k+2}$ without strategy change points. 

The argument is nearly identical to the one in Section~\ref{sec:no_approx}, but
instead of referencing Figure~\ref{fig:basiclb}, it
uses Lemma~\ref{lem:val_func_exp}. 

Thus, consider towards contradiction that $\epsilon<2^{-k}$ and that we have an $\epsilon$-optimal minimizer strategy with no strategy change points in $D=[x\cdot 2^{-k+1},(x+1)\cdot 2^{-k+1})$ for some integer $x$. 
Let $a=x\cdot 2^{-k+1}$ and $b=(x+1)\cdot 2^{-k+1}$. Thus, $D=[a,b)$.
 Since there are no strategy change points in $D$, minimizer has only three options for the duration of $D$: Always go to $v_{\ell}^{k-1}$, always go to $v_{r}^{k-1}$ or wait until the end of $D$ and then possible do something else. 
If he waits, the outcome, for play starting at time $a$ is $\geq \val(v_{\ell}^k,b)+2^{-k+1}$ (because at best the minimizer starts playing optimally thereafter), while $\val(v_{\ell}^k,a)=\val(v_{\ell}^k,b)+2^{-k+1}/2$, since in the first half, the slope is $0$ and in the last half it is $-1$, by Lemma~\ref{lem:val_func_exp}. The strategy is therefore not $\epsilon$-optimal, since $\epsilon<2^{-k}=2^{-k+1}/2$. 
Alternately, if we always go to state $v_{\ell}^{k-1}$ (and pay the cost of $2^{-k}$) or state $v_{r}^{k-1}$: 
We have that $|\val(v_{\ell}^{k-1},t)+2^{-k}-\val(v_{r}^{k-1},t)|=2^{-k}$, by Lemma~\ref{lem:val_func_exp} for $t\in\{a,b\}$, but which one is smaller differs.
Thus, at either $t=a$ or $t=b-(2^{-k}-\epsilon)/2$ (the adjustment is because there might be a change point at $b$) the outcome we get in $v_{\ell}^k$ differs from $\val(v_{\ell}^k,t)$ by at least $2^{-k}-(2^{-k}-\epsilon)/2=\frac{2^{-k}+\epsilon}{2}$ (atleast because minimizer need not play optimally after leaving $v_{\ell}^k$) - the differences in outcome between going to $v_{\ell}^{k-1}$ and $v_{r}^{k-1}$ changes linearly at all times because of Lemma~\ref{lem:val_func_exp}.
 The strategy is therefore not $\epsilon$-optimal, since $\epsilon<2^{-k}$.

The argument for the maximizer is symmetric and uses $v_{r}^k$ instead of $v_{\ell}^k$, but if the maximizer waits during $D$, the outcome for starting at time $a$ is $\leq \val(v_{\ell}^k,b)$, because $v_{r}^{k}$ has a holding rate of 0. Still $\val(v_{\ell}^k,a)=\val(v_{\ell}^k,b)+2^{-k+1}/2$ and the strategy is not $\epsilon$-optimal.

We get the following lemma.

%\begin{replemma}{lem:no_approx}
%There is a family of simple priced time games in which every $\epsilon$-optimal
%strategy with $\epsilon < 1/2^k$ uses $2^{k-1}$ change points.
%\end{replemma}
\noapproxlem*

\section{Upper bounds for trees and undirected graphs}
\label{app:pos}
In this section, we will argue that there are few event points (i.e. at most
polynomial many) in SPTGs that are either trees or undirected graphs. Recall
that the event point iteration algorithms runs in time $O(|E|(m+n\log n))$,
where $E$ is the set of event points. 

We will say that a line segment $L$ is covered by a line or line segment $L'$ if $L\subseteq L'$ and also extend the notion to  sets, i.e. a set $S$ is covered by a set $S'$ if each element $L\in S$ is covered by some element $L'\in S'$ (which may depend on $L$).

%\begin{replemma}{lem:trees}
%Consider a state $s$ which is the root of a tree with $k$ leaves, for some number $k$.
%Then, let $L_s$ be the line segments of $\val(s,t)$.
%There exists a set $L_k$ of $k$ lines that covers $L_s$.
%\end{replemma}
\treeslem*

\begin{proof}
The proof is by structural induction in the tree-structure.
Consider a leaf $s$. Either, it is a goal state or not. If it is then $\val(s,t)=0$ and otherwise $\val(s,t)=\infty$. In either case, they form a (horizontal) line, satisfying the statement.

Next, consider an inner node $v$ of the tree, with children $c_1,\dots, c_{\ell}$ such that $c_i$ can be covered by a set of $k_i$ lines, which is also a lower bound on the number of leaves below $c_i$. No matter if it is a maximizer or minimizer state, we have that $\val(v,t)$ is an upper/lower envelope of some functions, according to the value iteration algorithm. We can split the upper/lower envelope into some sets and then first apply the upper/lower envelope
on the sets on their own and then apply upper/lower envelope on the set of results for each set. This will still give the same result.

The sets $C=\{C_1,\dots,C_{\ell}\}$ are the value function and the corresponding line segments for each of the successors. Fix an $i$. It is clear that lower/upper envelope of $C_i$ is a piecewise linear function with at most as many pieces as the function $\val(c_i,t)$ and since $\val(c_i,t)$ could be covered with $k_i$ lines before, then the upper/lower envelope of the set can be covered by $k_i$ lines now (or fewer). This is because there is an added line segment for each event point and they are parallel. Thus, then a line segment is left it can never be reentered, because we must have entered a better line segment in-between and we could then just continue on that instead. Let the set that covers the upper/lower envelope of $C_i$ be $L_i$. As a side remark, we do not necessarily have that $L_i$ is the set that covers $\val(c_i,t)$, because the additional line segments might require us to change some lines.

If we have a set of lines $L_i$ that can cover each of the sets on their own, then the union $L$ of them can cover the lower/upper envelope, i.e. $L$ can cover $\val(v,t)$. Note that there are at atleast $k=\sum_{i=1}^{\ell}k_i$ leaves under $v$ (because there were at least $k_i$ leaves under $c_i$ and the leaves must be disjoint since the graph is a tree) and the size of $L$ is at most $k$. 
\end{proof}

\begin{theorem}
In a SPTG that forms a tree, there are at most $O(n^3)$ many event points
\end{theorem}
\begin{proof}
According to Lemma~\ref{lem:trees}, each state is made up of a subset of at most $n$ lines. There can be at most $\frac{n(n-1)}{2}=O(n^2)$ intersections of $n$ lines.
\end{proof}

The next lemma can especially be applied on undirected graphs.

%\begin{replemma}{lem:dags}
%Consider a graph such that for each pair of states $s,s'$ so that $s'$ is a maximizer state, if $(s,s')\in E$, then $(s',s)\in E$. Then there is a polynomial number of event points.
%\end{replemma}
\dagslem*

\begin{proof}
We have that minimizer will never, in an optimal strategy, from a state $s$ such that $\val(s,t)\neq \infty$ (it is easy to see that $\val(s,t')=\infty$ for some $t'$ iff $\val(s,t)=\infty$ for all $t$), take an edge to a maximizer state. 
This is because if he did so, the maximizer could just immediately move back to $s$
and there play would never reach goal. Thus, we can remove all such edges. Now,
the set of minimizer states forms a component, which we can solve first. Because
only 1 player has states in the component, it can be solved as a one player
game, and such has at most a polynomial number of event points. 

Next, there are two kinds of maximizer states $s'$, those with edges to another maximizer state $s''$ and those without. If $(s',s'')\in E$, then $(s'',s')\in E$ by the assumption of the graph. 
Thus, the maximizer can just move back and forth between the states forever and
thus $\val(s',t)=\val(s'',t)=\infty$. 
Otherwise, if the maximizer only have edges to minimizer states, we can see,
following the value iteration algorithm, that it is the upper envelope of a
polynomial number of line segments. Upper envelopes of line segments form a
Davenport-Schinzel sequence of order 3, which implies that, if you have $n$ line
segments, then the output consists of at most $2n\alpha(n)+O(n)$ line segments,
where $\alpha$ is the inverse of the Ackermann function - $\alpha$ is much
slower growing than e.g. $\log$. Thus, we see that we have at most a polynomial
number of event points.
\end{proof}

\section{\PSPACE upper bound for DAGs}
\label{app:pspace_upper_bound}

In this section, we show that DAGs can be solved in polynomial space, by arguing
that, for each state $v$, each end point of a line segment of $\val(v,t)$ is a
pair with at most polynomially many bits.

For a natural number $c$, we say that a fraction $x$ is {\em $c$-expressible} if the denominator $d$ of $x$ is such that $d\cdot k=c$ for some natural number $k$. 
We trivially see that if $p,q$ are $c$-expressible and $r$ $c'$-expressible, then $p-q$ and $p+q$ are $c$-expressible and $p\cdot r$ is $(c'\cdot c)$-expressible for all natural numbers $p,q,r,c,c'$. Also, if a number is $c$-expressible, then it is also $(c\cdot k)$-expressible for all natural numbers $k$.

\begin{lemma}\label{lem:R-express}
Consider a SPTG on a DAG of depth $h$ with integer holding rates. Let $R=\Pi_{v_1,v_2\in V\mid r(v_1)\neq r(v_2)}\abs{r(v_1)-r(v_2)}$. 
 Let $v$ be some state at depth $h_v$ and $(x,y)$ some end point of a line segment of $\val(v,t)$. 
If $y=\infty$, then $\val(v,t)=\infty$ and otherwise,
 if $y\neq \infty$, the numbers $x$ and $y$ are $R^{h-h_v}$-expressible.
\end{lemma}
\begin{proof}
We will show the statement using structural induction and consider how the value iteration algorithm computes the function.

First, for the leaves, which are at depth $h$. Either the leaf $v$ is a goal state or not. If so, then $\val(v,t)=0$ and otherwise, $\val(v,t)=\infty$. In either case, the statement is satisfied, in the former case, because the only end points are $(0,0),(1,0)$ and all those are $R^0=1$-expressible.

Next, consider a non-leaf node $v$ of height $h_v$, with successors $c_1,\dots,c_k$. By structural induction, we have that for each $i$, each end point $(x_i,y_i)$ of a line segment of $\val(v_i,x)$ are $R^{h-h_v-1}$. The -1 is because they are at a depth 1 larger.

First of all, the $\val(v,t)$ is either always $\infty$ or never $\infty$. 
To see this consider first that $v$ is a maximizer state. 
Either one of $c_1,\dots,c_k$ is such that $\val(c_i,t)=\infty$ or not. 
In the first case, $\val(v,t)=\infty$ and in the second, by definition of $\up$
and that each line segment has bounded end points, we have that $\val(v,t)$ is
finite. 
Otherwise, if $v$ is a minimizer state, then either all of $c_1,\dots,c_k$ is such that $\val(c_i,t)=\infty$ or not.
In the first case, $\val(v,t)=\infty$ and in the second, by definition of $\low$ and that each line segment has bounded end points, we have that $\val(v,t)$ is finite. 

To show the lemma we thus just need to consider the case where $y\neq \infty$.
In that case, each line segment of $\val(c_i,t)$ or the additional line segments in $L_i$ has a slope corresponding to minus some holding rate and starts in an end point of some line segment of $\val(c_i,t)$ for some $i$ and goes towards the left. Fix two such line segments, defined from end point $(x_i,y_i)$ and slope $-r_i$, for $i\in\{1,2\}$.
They can potentially generate a new end point of a line segment in $\val(v,t)$, at the intersection of the two line segments (the upper and lower envelope of the line segments can trivially not contain other new end points of line segments).
We consider the case when $x_1\geq x_2$ and the other case is symmetric.
We see that the line segment starting in $(x_1,y_1)$ is going through $(x_2,y_1+r_1 (x_1-x_2))$. We let $y'=y_1+r_1 (x_1-x_2)$. If $y'=y_2$, then the line segments intersect in $(x_2,y_2)$. We have by induction that $x_2,y_2$ are such that their denominator $d$ satisfies that $c\cdot d=R^{h-h_v-1}$ for some integer $c$ and thus satisfies the statement.
If $y'<y_2$ and $r_2\geq r_1$ or $y'>y_2$ and $r_2\leq r_1$, then the line segments will never meet. 
Finally, if $y'<y_2$ and $r_2<r_1$ or $y'>y_2$ and $r_2>r_1$, then line segments intersect in $(x,y)$ where
$x=\abs{y'-y_2}/\abs{r_1-r_2}$ and $y=(x_2-x)r_2+y_2(=(x_2-x)r_1+y')$.

First, observe that $y'=y_1+r_1 (x_1-x_2)$  is $R^{h-h_v-1}$-expressible, 
because each of $y_1,x_1,x_2$ are and $r_1$ is some integer and thus $1$-expressible.
Next, $x=\abs{y'-y_2}/\abs{r_1-r_2}$ is $R^{h-h_v}$-expressible, because $y',y$ are $R^{h-h_v-1}$-expressible and $1/\abs{r(v_1)-r(v_2)}$ is $R$-expressible. 
Finally, $y=(x_2-x)r_2+y_2$ is $R^{h-h_v}$-expressible, because $x_2,y_2$ are $R^{h-h_v-1}$-expressible and thus $R^{h-h_v}$-expressible, $x$ is  $R^{h-h_v}$-expressible and $r_2$ is some integer and thus $1$-expressible.
\end{proof}

\begin{theorem}
Consider an SPTG on a DAG with integer holding rates. Then DecisionSPTG is in \PSPACE.
\end{theorem}
\begin{proof}
We will first argue that for each event point $t'$ and state $v$, we have that $\val(v,t')$ is $R^{h}$-expressible.
Fix a state $v$ and a event point $t'$.
An event point is the x-coordinate of an end point of a line segment of $\val(v',t)$ and, by Lemma~\ref{lem:R-express}, event points are therefore $R^{h}$-expressible.
The number $\val(v,t')$ is on some line segment of $\val(v,t)$, starting from some end point of a line segment $(x,y)$ which is $R^h$-expressible and having slope $-r$ for some holding rate $r$. Thus, $x\geq t'$ and hence $\val(v,t')=(t'-x)r+y$. Note that $t',x,y$ are $R^h$-expressible and $r$, being some holding rate, is $1$-expressible. Thus, $\val(v,t')$ is $R^h$-expressible.

The event point iteration algorithm, given an event point $t'$ and $\val(v,t')$ for all $v$, finds the next smaller event point $t''$ and $\val(v,t'')$ for all $v$. and finally outputs all of these numbers as the value function.
To solve the DecisionSPTG problem with $t^*$ as input time and $v$ as input state (i.e. we want to find $\val(v,t^*)$), we then use a variant of the event point iteration algorithm that iterates over the event points, but simply deletes the values from previous iterations, until the iteration in which we have the smallest event point $t'\geq t^*$ and then finds $\val(v,t^*)$ from that. 

We have that $\val(v',t')$ are $R^h$-expressible as shown above. Let $M$ be the largest holding rate (which requires linear space).  For obvious reasons, the nominator of $t'$ is smaller than its denominator (or equal, but only in case $t'=1$). By the event point iteration algorithm, we see that $\val(v,0)\leq \val(v,1)+M$ for all $v$. Also, $\val(v,0)\geq \val(v,t)$ for all $t$.
And finally,  $\val(v,1)\leq n\cdot W$ (unless $\val(v,1)=\infty$), where $W$ is the biggest weight, since it is the cost of some acyclic path in the graph from $v$ to a goal node. Hence, the nominator of $\val(v,t')$ 
Hence, the nominator of $\val(v,t')$ can at most be $n\cdot W+M$ times the denominator (which was $R^h$-expressible).
Thus, each of the $2n+2$ numbers (i.e. nominator and denominator) is bounded by $(n\cdot W+M)R^h$. 
We see, by definition, that $R$ has at most $n^2$ factors, each of which are at most $M$. Hence, we need at most $h\cdot n^2\cdot \log_2(M)$ bits to write down $R^h$ and thus any denominator of a $R^h$-expressible number in each iteration.
Each of the $2n+2$ numbers thus takes up at most $\log_2(n\cdot W+M)+h\cdot n^2\cdot \log_2(M)$ many bits, which is polynomial in the input size. 

In the last iteration (i.e. the one where we output $\val(v,t^*)$), we have an event point $t'\geq t^*$ and the next event point $t''<t^*$. We have that $\val(v,t^*)$ is then $\val(v,t')+r\cdot (t'-t^*)$ for the holding rate $r$ of that line segement of $\val(v,t)$. We see that this computation can be done in polynomial space, since $\val(v,t'),t'$ fits in polynomial space and $r,t^*$, being explicitly in the input, takes at most linear space.

In conclusion, the above described variant of the event point iteration algorithm runs in \PSPACE and solves DecisionSPTG.
\end{proof}

\section{A helpful lemma}

Before giving our lower bound constructions, we state following useful lemma. 
The intuition behind it was also an important part of the proof in \cite{BLMR06}
that one clock priced timed games have a value.

\begin{lemma}\label{lem:no wait}
In a minimizer state $s$ with $r(s)=\max_v r(v)$, the minimizer can always avoid waiting, i.e. 
$$\val(s,t)=\min_{v\in V\mid (s,v)\in E}(\val(v,t)+c((s,v))).$$

Similarly, in a maximizer state with holding rate $0$, the maximizer can always avoid waiting, i.e.,
$$\val(s,t)=\max_{v\in V\mid (s,v)\in E}(\val(v,t)+c((s,v))).$$
\end{lemma}
\begin{proof}
We consider the case with a maximizer state of holding rate $0$ and the argument for
minimizer states (with holding rate $r(s)=\max_v r(v)$) is similar.

The argument is based on the two algorithms described in Appendix~\ref{app:known_algo}. 
We have from the value iteration algorithm (Algorithm~\ref{alg:val_it}) that
$\val(s,t)$ is the upper envelope of a set of functions, some that corresponds
to waiting and some that corresponds to going to a successor. However, the
waiting line segments cannot be above the value function of the successor that
spawned it. This is because each line segment has slope at least 0 (by 
Algorithm~\ref{alg:ev_it}). 
\end{proof}

\section{Encoding formulas}
This section is intended to serve as a link between our exponentially lower bound, in the previous section and our \NP and \coNP hardness proof in the next. 
In essence, we will use our exponential lower bound to encode the variables in SAT/Tautology and then later TQBF.

\paragraph{Assignment time interval.}
For the set of booleans $B=\{b_1,\dots,b_n\}$, an assignment $A$ is a map from $B$ to $\{0,1\}$ (or false and true).
For an assignment $A$, we define an interval of times $T^A$, where $t\in T^A$ iff, for all $i\in [n]$, the $i$-th bit of $t$ after the comma, $t_i$, is $A(b_i)$ and there exists some $j>n$, such that $t_j\neq 1$ (this latter requirement is for formal reasons, since e.g. $0.11...=1$ in binary by definition of the reals). Note that each time $t\neq 1$ is mapped to an assignment ($t=1$ is not mapped to any assignment). 
\begin{itemize}[leftmargin=0.3cm]
\itemsep1mm
\item {\bf $t^A_s$.} We let $t^A_s$ be the first $t\in T^A$, i.e. $t\in T^A$ and $t^A_{s,i}=A(b_i)$, for $i\in [n]$ and $t_j=0$ for $j\geq n+1$.
\item {\bf $t^A_m$.} We let $t^A_m$ be the middle $t\in T^A$, i.e. $t\in T^A$ and $t^A_{s,i}=A(b_i)$, for $i\in [n]$, $t_{n+1}=1$ and $t_j=0$ for $j\geq n+1$.
\item {\bf $t^A_e$.} We let $t^A_e$ be the last $t\in T^A$, i.e. $t\in T^A$ and $t^A_{s,i}=A(b_i)$, for $i\in [n]$ and $t_j=1$ for $j\geq n+1$ (it is equal to $t^{A'}_s$, where $A'$ is the ``next'' assignment, by definition of reals).
\end{itemize}

\paragraph{Definition of function encoding.}
We will encode functions using two numbers, a {\em value} $v$ and a {\em offset}
$v'$. We have two encodings, {\em straight encoding} and {\em reverse encoding}
(that depends on $v,v'$ and the set of variables it is over, the latter because
it changes the duration of an assignment). Straight is used for exists, such as
SAT and in exists alternations of TQBF and reverse encoding is used for for all,
such as Tautology and for all alternations of TQBF. We will then convert from
one to the other as a part of our \PSPACE-hardness proof.

For a boolean function $F$ (in particular (quantified) boolean formulas) and an
assignment $A$ to its variables, we write $F(A)$ for the boolean the function
evaluates to when the variables are assigned according to assignment $A$.

\begin{definition}[Straight encoding] 
A state $s$ encodes a boolean function $F$ under straight encoding iff 
\begin{enumerate}[label=\textnormal{(S\arabic*)}]
\itemsep2mm
\item \label{str1} $\forall t$ we have that $$\val(s,t)\in [v-t/2,v+v'-t/2];$$
\item \label{str2} $\forall A$ s.t. $F(A)$ is false, we have that 
		$$\forall t\in t^A:\val(s,t)=v-t/2.$$
\item \label{str3} $\forall A$ s.t. $F(A)$ is true, we have that \ $\exists t\in t^A:$
	\begin{align*}
		\val(s,t) & = v+v'-t/2; \quad \text{and} \\
		\val(s,t_s^A) & =\val(s,t_e^A)=v-t/2;
	\end{align*}
\end{enumerate}
\end{definition}

\begin{definition}[Reverse encoding] 
A state $s$ encodes a boolean function $F$ under reverse encoding iff 
\begin{enumerate}[label=\textnormal{(R\arabic*)}]
	\item \label{rev1} $\forall t$ we have that $$\val(s,t)\in [v-v'-t/2,v-t/2];$$
\item \label{rev2} $\forall A$ s.t. $F(A)$ is true, we have that 
	$$\forall t\in t^A:\val(s,t)=v-t/2;$$
\item \label{rev3} $\forall A$ s.t. $F(A)$ is false, we have that $\exists t\in t^A:$ 
	\begin{align*}
		\val(s,t)=v-v'-t/2; \quad \text{and} \\
		\val(s,t_s^A)=\val(s,t_e^A)=v-t/2;
	\end{align*}
\end{enumerate}
\end{definition}

Since the encodings are so similar, we will at times be able to prove results for both variants at the same time. We will, for $i\in\{1,2,3\}$, then use (en$i$) to refer to (str$i$) for the proof for straight encoding and to (rev$i$) for reverse encoding.

\subsection{Encoding booleans}
\label{app:enc_bools}

In this sub-section we will show how to use our lower bound family to encode
booleans, i.e. $v_i$ and $\neg v_i$ for each $i$, in our encodings.

\paragraph{Step 1 of using our game to encode a variable.}
To encode a variable, $v_i$, we first construct a game according to our lower
bound family from the previous section, with $i$ levels. We then add a new state
$s_i$ (which will encode $v_i$ - whether it is minimizer state of holding rate 1 or a
maximizer state of holding rate 0 does not matter) which has an edge $e$ to $v_{r}^{i}$.
The cost of $e$, $c(e)$, is $ 1/4-2^{-i-2}+2^{-i-1}$. This means that
$\val(s_i,t)$ intersect, according to Lemma~\ref{lem:val_func_exp}, the function
$(\val(v_{\ell}^i)+1/4-2^{-i-2}+2\cdot 2^{-i-1})$ $2^{i+1}$ times on the line
$L$ that starts at $5/4-2^{-i-2}$ and has slope $-1/2$. The value $\val(s_i,t)$
has some similarities with what we want, when we encode $v_i$ in (straight or
reverse) encoding with $v=5/4-2^{-i-2}$ and $v'=2^{-i-2}$. In particular,
\begin{enumerate}[leftmargin=0.3cm]
\itemsep2mm
\item There are $2^{i+1}$ durations (the first and last are of length $2^{-i-2}$,
the rest of length $2^{-i-1}$ each); 
\item For all times $t$ in every odd duration, $$v-t/2\leq \val(s,t)\leq v+v'-t/2,$$ 
and for some\footnote{at the start
of the first and in the middle of the remaining durations} $t$ in each such
duration $$\val(s,t)= v+v'-t/2;$$
\item For each $t$ in an even duration
$$v-v'-t/2\leq \val(s,t)\leq v-t/2,$$ and for some\footnote{at the end of the last
and in the middle of the remaining durations} $t$ in each such duration
$$\val(s,t)=v-v'-t/2.$$ 
\end{enumerate}
To encode $\neg v_i$ (i.e., a variable which is true iff
$v_i$ is false), we can then first use a state $\hat{s}_i$, which has an edge
$e'$ to $v_{\ell}^{i}$ of cost $c(e')=1/4-2^{-i-2}+2\cdot 2^{-i}$. Such a
construction have similar properties to that of $s_i$ (but the properties of odd
and even durations are exchanged). The constructed states $s_i$ and $\hat{s}_i$
have values similar to what we want in our encoding of $v_i$ and $\neg v_i$
(with $v=5/4+2^{-i-2}$ and $v'=2^{-i-2}$), but there are three differences: 
\begin{enumerate}[label=\textnormal{(D\arabic*)}]
\item We start and end in the middle of a duration (and it is not equal to
	$v-t/2$ at the start of the first duration or end of the last duration); \label{issue1}
\item There are two times too many durations;  \label{issue2}
\item In straight encoding the value should be equal to $v-t/2$ if the variable
	is false (instead of below it) and in reverse encoding it should be equal to
	$v-t/2$ (instead of above it) when the variable is true. \label{issue3}
\end{enumerate}

We next deal with these three issues, first dealing with \ref{issue1} and \ref{issue2} by \emph{shifting}.

\paragraph{Shifting.}
To deal with \ref{issue1} and \ref{issue2}, we would like to have a state $s_i'$ such that $\val(s_i',t)=\val(s_i,t/2+1/2-2^{-i-1})$ for $t\in [0,1]$ (we use $t/2$ to ensure that we do not consider the function $\val(s_i,t)$ outside the $[0,1]$ interval for $t$ that we have already considered). 
Note that the last duration (which is the one cut in half) is of length $2^{-i-1}$, which is why $2^{-i-1}$ appears here.
Also, we keep only half of the durations (because of $t/2$).
By definition, $\val(s_i',t)$ has value $1$ at $t=0$ (because $\val(s_i,1/2-2^{-i-1})$, is equal to $\val(L,1/2-2^{-i-1})=5/4-2^{-i-2}-(1/2-2^{-i-1})/2=1$) and $3/4$ at $t=1$ (because $\val(s_i,1-2^{-i-1})$, is equal to $\val(L,1/2-2^{-i-1})=5/4-2^{-i-2}-(1-2^{-i-1})/2=3/4$), also, there are $2^i$ many durations in a game of half the length, and it is above $1-t/4$ half the time and below it the other half.

To do so, we construct a modified game according to \cite[Lemma 4.5-4.8]{HIM13}. The lemmas can be used to construct a game $G_1$ for the interval
$[2^{-i-2},2^{-i-2}+1/2]$ (which is the interval we care about), see
Figure~\ref{fig:extract-time-period} for an illustration.

\begin{figure}
\centering
\scalebox{0.8}{
\pgfkeys{/pgfplots/myaxis/.style={xmin=0,ymin=0,xmax=1.1,ymax=1.1, 
								  samples=2, % straight lines so no need for more than 2
								  axis lines=center,
						  		  xtick = {1},
								  extra x ticks = {0},
								  ymajorticks=false,
								  %ytick = {1},
								  %extra y ticks = {0},
								  %minor tick num=1,
								  %axis lines = middle,
								  xlabel=$\mathsf{time}$,
								  label style = {at={(ticklabel* cs:1.15)}, anchor=east, font=\large},
								  tick label style={font=\large},
								  %ylabel=$y$,
								  %width=13.5cm, height=5cm,
								  %grid=both
								  %xticklabel=\empty,yticklabel=\empty,
						  		  tick label style={
										/pgf/number format/.cd,
										std,
										precision=3,
										/tikz/.cd
									}
						  		  }}

\begin{tikzpicture}
\begin{axis}[myaxis, extra x ticks={0.375,0.875}]
  % blue
  \addplot[thick][domain=0:0.25] (x,1-x);
  \addplot[thick][domain=0.75:1] (x,0.5);
  \addplot[thick][domain=0.25:0.5] (x,0.75);
  \addplot[thick][domain=0.5:0.75] (x,1.25-x);
  % red
  \addplot[thick][domain=0.25:0.5] (x,1.125-x);
  \addplot[thick][domain=0.5:0.75] (x,0.625);
  \addplot[thick][domain=0:0.25] (x,0.875);
  \addplot[thick][domain=0.75:1] (x,1.375-x);
  %\draw[dashed] (axis cs: 0.125,0.125) rectangle (axis cs: 0.875,0.875);
  \draw[dashed] (axis cs: 0.375,0) rectangle (axis cs: 0.875,1);
\end{axis}
\end{tikzpicture}
}
\caption{Extracting a time period from a longer game.}
\label{fig:extract-time-period}
\end{figure}
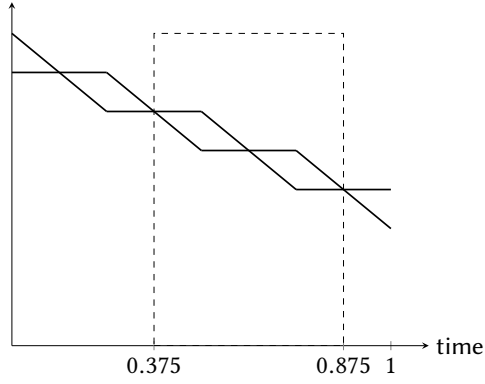

The resulting SPTG then has a state that has value function equal to $s_i'$ as wanted. 
The game is quite similar to the original construction, but all nodes have holding rate
either 0 or $1/2$ (basically, because the interval it is constructed over is of
length $1/2$ and all holding rates were 0 and 1 before) and some additional states and
edges have been added\footnote{for each original minimizer state $s_{\ell}^i$
	there is a new maximizer state $v_{\ell}^i$ with holding rate 1 and an edge to it of
	cost $0$ and then an edge from $v_{\ell}^i$ to some new goal state
	$g_{\ell}^i$ of cost $c_{\ell,i}=\val(s_{\ell}^i,1-2^{-i-2})$ in the
original game. For each original maximizer state $s_{r}^i$, for $i\geq 0$ there
is an new goal state $g_{r}^i$ and an edge to it from $s_{r}^i$ of cost
$c_{r,i}=\val(s_{r}^i,1-2^{-i-2})$ in the original game}. We would prefer to use
holding rates $0$ and $1$, so we use a classic trick from game theory and change the
currency of the output to one of half as much value (such tricks are also used
in \cite{HIM13}). This causes the holding rates, the cost and the values (at all times),
when expressed using the new currency, to double. Let $G_2$ be that game (which
has holding rates in $\{0,1\}$). Let $s_i''$ be the state in $G_2$ corresponding to
$s_i$ (or equally to $s_i'$). 

We see that $\val(s_i'',t)$ is such that (1)~there are $2^{i}$ durations (each of length $2^{-i}$ each); (2)~for each time $t$ in every {\em even} duration, $2-t/2\leq \val(s,t)\leq 2+2^{-i-2}-t/2$, and for the middle $t$ in each such duration $\val(s,t)= 2+2^{-i-2}-t/2$ and (3)~for each $t$ in an {\em odd} duration  $2-2^{-i-2}-t/2\leq \val(s,t)\leq 2-t/2$ and for the middle $t$ in each such duration $2-2^{-i-2}-t/2= \val(s,t)$. Note also that slope alternates in the middle of the durations and that $\val(s,t)=2-t/2$ at the start and end of each duration.

We would like to eliminate the additional states. We will show two lemmas that together will do so.

\begin{lemma}
Consider the game $G_2$ above.
Consider some $i\geq 1$ and let $s=s_{\ell}^i$ and $v_{\ell}^i$. If  $\val(s,t)=\val(v,t)$, then $\val(s,1)=\val(v,1)$ and $\val(s,t)=1-t+\val(v,1)$.

Also, let $s'=s_{r}^i$ and $v'=g_{r}^i$. If $\val(s',t)=\val(v',t)+2\cdot c_{r,i}$, then $\val(s',t)=2\cdot c_{r,i}=\val(s',1)$
\end{lemma}
\begin{proof}
Observe that $\val(v,t)=c_{\ell,i}+1-t$, since the maximizer will wait until time 1 and go to goal.
Therefore, $\val(s,t)=\val(v,t)=c_{\ell,i}+1-t=1-t+\val(v,1)$. On the other hand, if the minimizer waits in $s$ until time 1 and then move as at time 1, we have that $\val(s,t)\leq \val(s,1)+1-t$. But, the minimizer can go to $v$ at time $1$, so $\val(s,1)\leq \val(v,1)$. Therefore, since $\val(s,t)=1-t+\val(v,1)$, we see that $\val(s,1)=\val(v,1)$.

On the other hand, if $\val(s',t)=\val(v',t)+2c_{r,i}$ (recall that we multiplied all holding rates and costs by 2 to get $G_2$, which explains why there is a 2 here). We have that $\val(v',t)=\val(v',1)=0$ since it is a goal state and thus,  $\val(s',t)=c_{r,i}$. On the other hand, the maximizer can wait until time 1 in $s'$ (costing 0) and then do whatever he would do according to an optimal strategy at time 1. Thus, $\val(s',t)\geq \val(s',1)$. One of the options for the maximizer at time 1 is to go to $v'$ and thus $\val(s',1)\geq 2\cdot c_{r,i}$. But, then $\val(s',t)=2\cdot c_{r,i}=\val(s',1)$.
\end{proof}

The lemma says that if a strategy ever goes from $s$ to $v$ or from $s'$ to $v'$, then we can assume that it first wait until time 1 and then do so. We will next argue that (besides $v_r^0$, which must use this new edge, but on the other hand, we can just change the cost of the edge to $s_{\ell}^0$ instead of adding a new state) that at time 1, there are equally good options to going to the additional states. 

\begin{lemma}
For all $k\geq 0$, in $G_2$, we have that \[
\val(s_{\ell}^{k},1)=\begin{cases} 
\val(s_{\ell}^{k-1},1)+2^{1-k} & \text{If }k\geq 2\\
\val(s_{r}^{0},1) & \text{If }k=1\\
0 & \text{If }k=0
\end{cases}
\]
Also, 
\[
\val(s_{r}^{k},1)=\begin{cases}
\val(s_{r}^{k-1},1) & \text{If }k\geq 2\\
\val(s_{\ell}^{0},1)+1/2 & \text{If }k=1\\
2\cdot c_{r,0} & \text{If }k=0
\end{cases}
\]
\end{lemma}
\begin{proof}
The state $s_{\ell}^0$ is a goal state, thus $\val(s_{\ell}^0,t)=0$. 
In the remainder of the proof, let $t=1-2^{-i-1}$.
In state $s_{r}^0$, the maximizer can either go to goal with a cost of 0 or of
$2\cdot c_{r,0}$.
The number $c_{r,0}$ was $\val(s_{r}^0,t)$ in $G$. The optimal choice in
$s_{r}^0$ in $G$ was to wait to time 1 and go to goal. We have, in the original
game, that $\val(s_{r}^0,t)=1-t=2^{-i-1}$. Clearly, the latter is preferable.
Thus, $\val(s_{r}^{0},1)=2\cdot c_{r,0}=2^{-i}$.

All nodes $s_{\ell}^k$ and $s_{r}^k$ (in both $G$ and $G_2$), for $k\geq 1$ are
either minimizer states of holding rate 1 or maximizer states of holding rate 0 and thus
satisfies Lemma~\ref{lem:no wait}. We will therefore apply that lemma with no
more references for those states in this proof.

We will show using induction in $k$, for all $1\leq k\leq i$, that
$c_{\ell,k}=2^{-i-1}+\sum_{j=2}^k 2^{-j}=1/2-2^{-k}+2^{-i-1}<1/2$ (in
particular, for $i=k$, we have that $c_{\ell,i}=1/2-2^{-i-1}$ and this is the
top state), and $c_{r,k}=1/2$.

First, the base case.
For $s_{\ell}^1$ in $G$, we have that 
\begin{alignat*}{2}
c_{\ell,1} = & \val(G,s_{\ell}^1,t)\\
		   = & \min(\val(G,s_{\ell}^0,t)+1/2, \val(G,s_{r}^0,t)) \\
		   = & \min(1/2,2^{-i-1})=2^{-i-1}.
\end{alignat*}

For $s_{r}^1$ in $G$, we get the same expression, but with $\max$ replacing $\min$ and thus $c_{r,1}=1/2$.

Next, induction case.
For $s_{\ell}^k$ in $G$, we have that 

\begin{alignat*}{2}
c_{\ell,k}=& \val(G,s_{\ell}^k,t)\\
		  =& \min(2^{-k}+\val(G,s_{\ell}^{k-1},t),\val(G,s_{r}^{k-1},t)) \\
		  =& \min(2^{-k}+2^{-i-1}+\sum_{j=2}^{k-1} 2^{-j},1/2) \\
		  =& \ 2^{-i-1}+\sum_{j=2}^{k} 2^{-j}.
\end{alignat*}

For $s_{r}^k$ in $G$, we again have get the same expression, but with $\max$ replacing $\min$ and thus $c_{r,k}=1/2$.

Next, we will, again using induction in $k$, for all $k\geq 0$, that $\val(G_2,s_{\ell}^k,1)=2c_{\ell,k}$ and $\val(G_2,s_{r}^k,1)=2c_{r,k}$ and that there is an option that is equally good as going to $v_{\ell}^k$ and $g_r^k$ in $s_{\ell}^k$ and $s_r^k$ respectively.

First, the base case, for $k=1$.
For $s_{\ell}^1$ in $G_2$, we have that 

\begin{alignat*}{2}
	\val(G_2,s_{\ell}^1,1)= & \min( & \val(G_2,s_{\ell}^0,1)+1,\\
							&	    & \val(G_2,s_{r}^0,1), \\
						    &       & 2c_{\ell,1} + \val(G_2,v_{\ell}^1,1))\\
						  = & \min( & 1,2^{-i},2^{-i}) =2^{-i}=2c_{\ell,1}.
\end{alignat*}

Note that $\val(G_2,s_{\ell}^1,1)=\val(G_2,s_{r}^0,1)$.
For $s_{r}^1$ in $G_2$, we have get the same expression, but with $\max$ replacing $\min$ and thus $\val(G_2,s_{r}^1,1)=2c_{r,1}$. Also, $\val(G_2,s_{r}^1,1)=1+\val(G_2,s_{\ell}^0,1)$.

Next, the induction case, for $k\geq 2$.
For $s_{\ell}^k$ in $G_2$, we have that 
\begin{alignat*}{2}
	\val(G_2,s_{\ell}^k,1) ={}& \min({}&\val(G_2,s_{\ell}^0,1)+2^{1-k}, \\
							  &      {}&\val(G_2,s_{r}^0,1),\\
							  &      {}&2c_{\ell,k}+\val(G_2,v_{\ell}^k,1)) \\
  					     ={}& \min({}&2c_{\ell,k-1}+2^{1-k},1,2c_{\ell,k}) \\
						 ={}& \min({}&2c_{\ell,k},1,2c_{\ell,k}) \\
						 ={}& {}&2c_{\ell,k}.
\end{alignat*}
Note that $\val(G_2,s_{\ell}^k,1)=\val(G_2,s_{\ell}^{k-1},1)+2^{1-k}$.
For $s_{r}^k$ in $G_2$, we have get the same expression, but with $\max$ replacing $\min$ and thus $\val(G_2,s_{r}^k,1)=1=2c_{r,k}$. Also, $\val(G_2,s_{r}^k,1)=\val(G_2,s_{r}^{k-1},1)$.
\end{proof}

We can therefore remove all the states $g_{\ell}^k,v_{\ell}^k,v_r^k$ from $G_2$ and change the cost of the edge $e=(v_{r}^0,v_{\ell}^0)$ so that $c(e)=2^{-i}$ and still have that all states have the same value as before we did so.

An identical sequence of transformations can deal with \ref{issue1} and \ref{issue2} for $\neg v_i$.
 
\paragraph{Bounding the booleans.}
Finally, we deal with \ref{issue3}. 
To do so, we add a new maximizer state $L$, which has only a single edge $e$. 
The holding rate of $L$ is $r(L)=1/2$. The edge $e$ goes to a new goal state and has cost $3/2$. 

The optimal strategy for the maximizer when in $L$ is to wait until time 1 and go to goal. 
Thus, the value is $\val(L,t)=(1-t)/2+3/2=2-t/2$.
We then add a state $s_i^*$, with an edge to $L$ and an edge to $s_i''$.
In each case, we will have that $s_i^*$ is an encoding with $v=2$ and $v'= 2^{-i-2}$.
\begin{enumerate}[leftmargin=0.4cm]
\itemsep2mm
\item {\bf Straight encoding:} In this case, $s_i^*$ should be a maximizer state of holding rate $0$.
The state $s_i^*$ is then a straight encoding of $v_i$ (note that the value $\val(s_i^*,t)\geq v$, because of $L$ and $s_i^*$ being a maximizer state).
\item {\bf Reverse encoding:} In this case, $s_i^*$ should be a minimizer state of holding rate $1$.
The state $s_i^*$ is then a reverse encoding of $v_i$ (note that the value $\val(s_i^*,t)\leq v$, because of $L$ and $s_i^*$ being a minimizer state).
\end{enumerate}
The states are such that~(en3) is satisfied by the midpoint of the duration in
which the variable is true for straight and false for reverse, also the slope of the function in such a duration is
$t$ in the first half and $-t$ in the second half.

Again, we can deal with \ref{issue3} for $\neg v_i$ in the same way.

We get the following lemma:
\begin{lemma}
For each variable $v_i$ or $\neg v_i$, we can, using $5(i+1)$ states, construct a state that is a (straight or reverse) encoding of it with $v=2$ and $v'=2^{-i-2}$
\end{lemma}

\subsection{Formula encoding}
\label{sec:form_enc}
In this section, we will show how we encode a boolean formula~$F$ over $n$ variables in our construction. We assume that  De Morgan's laws have been applied repeatedly so that all negations are on variables only. In the previous section section, we gave an encoding for variables and their negation, and we thus need only show how we encode ANDs and ORs. 

\paragraph{\bf ANDs and ORs.}
We give a recursive implementation of ORs and ANDs, using the same encoding
for both types of gate in both straight and reverse encoding. Consider some AND
or OR over sub-formulas $F_1,F_2,\dots F_k$, for some $k\geq 1$. We can
recursively construct a game $A_i$, such that a state $s_i$ in it encodes $F_i$
for each $i$ (to ensure that our construction have certain properties, in
particular constant tree-width, the game for $F_i$ has no state in common with
the game for $F_j$ for $i\neq j$). In common for both ANDs and ORs, we add a
state $s$ that has an edge $e_i$ to $s_i$ for each $i$, of cost $c(e_i)=0$.

\smallskip

\begin{enumerate}[leftmargin=0.4cm]
\itemsep1mm
\item {\bf AND:} In this case, $s$ is a minimizer state of holding rate $1$.
\item {\bf OR:} In this case, $s$ is a maximizer state of holding rate $0$.
\end{enumerate}

\smallskip

The state $s$ has value $\val(s,t)$ and is close to satisfying our properties
for straight/reverse encoding. In particular, it satisfies \ref{str1} and
\ref{rev1} - for some bounds, which is trivial - and \ref{str2}/\ref{rev2} for
$v=2$. The latter comes from the fact that all variables $v_i/\neg v_i$ satisfies (en2) with $v=2$
((en2) is independent of $v'$). 
It then follows from Lemma~\ref{lem:no wait} that for any AND or OR directly of
variables that (en2) is satisfied for (straight or reverse) encoding with $v=2$
and it then follows recursively for all formulas.

However, $\val(s,t)$ will not necessarily satisfy requirement \ref{str3}/\ref{rev3} of
straight/reverse encoding, because 
different variables use different values of $v'$. 
We deal with it for the full formula by using a detector stater; the nodes
corresponding to internal ANDs and ORs in the formula will not (necessarily) be
straight or reverse encodings.

\paragraph{Detector state.}
The role of the detector state $c$ is to convert the input state (here the
state $s$ as above for our full formula $F$) to the right format. This will
allow us to detect the truthfulness of the formula according to the 
game value at a specific time.
As mentioned above, the issue with $s$ is that $\val(s,t)+t/2$ can vary a lot,
but we need a state that encodes $F$, in either straight or reverse encoding.
To do this, we make $c$ encode $(F\wedge (v_n\vee \neg v_n))$ for
straight encoding and $(F\vee (v_n\wedge \neg v_n))$ for reverse encoding. Note
that $(v_n\vee \neg v_n)$ is true for any assignment, and $(v_n\wedge \neg v_n)$
is false for all assignments; thus both these boolean formulas are equal to~$F$.

\begin{lemma}\label{lem:detector}
The state $c$ (whether we consider straight or reverse encoding) encodes the formula $F$, with $v=2$ and $v'=2^{-n-2}$. Also requirement (en3) is satisfied precisely by $t^A_m$ for any $A$ for which $F(A)$ is true/false
\end{lemma}
\begin{proof}
By the same argument as for $s$, $c$ satisfies requirement (en1) and (en2).

We will give the following argument for straight encoding (resp. reverse encoding):
\begin{claim}
If $F$ is true (resp. false) for some assignment $A$, then $\val(s,t^A_m)\geq 2+2^{-n-2}$ (resp. $\val(s,t^A_m)\leq 2-2^{-n-2}$).
\end{claim}
\begin{proof}
We do so recursively.
\begin{itemize}[leftmargin=0.3cm]
\itemsep=2mm
\item It is true for $v_n$ and $\neg v_n$ straightforwardly (because the duration of $t^A_m$ is exactly the duration for the corresponding state to $v_n$ resp. $\neg v_n$ and variables reach their maximum in the middle).
\item It is true for each state $s$ corresponding to another variable, because $\val(s,t)+t/2$ increases (resp. decreases) from $v'$ (resp. $v$) at the start of the duration for $s$, until the middle and then decrease back to $v'$ (resp. $v'$) at the end. The first half of $\val(s,t)+t/2$ has slope $t/2$ (resp. $-t/2$) and the last half slope $-t/2$ (resp. $t/2$). For any $i<j$, the duration of the variable $v_{i}$ is $2^{j-i}$ times the duration of $v_j$ and the duration of each variable $v_i$ starts at the start of a duration for $v_j$ and ends at the end of a duration of $v_j$. Thus, if $v_i$ and $v_j$ are true in $A$, for $i<j$, then $\val(s_i^*,t)\geq \val(s_j^*,t)$ (resp. $\val(s_i^*,t)\leq \val(s_j^*,t)$) for all $t\in T^A$. It then follows by considering $j=n$.
\item It is true for each state $v$ that encodes an OR over $F_1,\dots, F_n$ (with state $f_i$ encoding $F_i$).
It is then a maximizer state of holding rate $0$ and thus $\val(v,t)=\max_i(\val(f_i,t))$ (by Lemma~\ref{lem:no wait}). 
If the OR is true (resp. false) for $A$ then it follows since for some (resp. each) $F_i$, we have $F_i(A)$ is true (resp. false) and it is then true for $f_i$ and thus for $v$ in turn.
\item It is true for each state $v$ that encodes an AND, similar to OR.
\end{itemize}
\end{proof}

State $c$ satisfies (str3) for straight encoding:
Consider the node $v$ that encodes the OR outside $F$ in the formula $(F\wedge (v_n\vee \neg v_n))$. It is then a maximizer state of holding rate $0$ and thus $\val(v,t)=\max(\val(s_n^*,t),\val(\hat{s}_n^*,t))$ (by Lemma~\ref{lem:no wait}). 
Observe that $\val(s_i^*,t),\val(\hat{s}_i^*,t)\leq 2+2^{-n-2}$, since they encode the variables $v_n/\neq v_n$ in a straight encoding (with $v=2$ and $v'=2^{n-2}$) and $\val(s_i^*,t^A_m)=2+2^{-n-2}$ or $\val(\hat{s}_i^*,t^A_m)=2+2^{-n-2}$ (because either $v_n$ or $\neq v_n$ is true in each assignment). Also, \[
\val(s_i^*,t^A_s)=\val(\hat{s}_i^*,t^A_s)=\val(s_i^*,t^A_e)=\val(\hat{s}_i^*,t^A_e)=2-t/2
\]
by straight encoding and thus 
\begin{alignat*}{2}
	\val(v,t^A_s)= & \max(\val(s_i^*,t^A_s),\val(\hat{s}_i^*,t^A_s))\\
				 = & \ 2-t/2\\
				 = & \max(\val(s_i^*,t^A_e),\val(\hat{s}_i^*,t^A_e))\\
				 = & \val(v,t^A_e).
\end{alignat*}
for all $A$ (using Lemma~\ref{lem:no wait}).
Since $c$ is a minimizer state of holding rate $1$ (being an AND), we then have that 
$$\val(c,t)=\min(\val(s,t),\val(v,t))$$ (by Lemma~\ref{lem:no wait}).
We get that $c$ satisfies (str3), since $\val(v,t^A_m)= 2+2^{-n-2}-t/2$ (and this is not true for other $t\in T^A$) and $\val(s,t^A_m)\geq 2+2^{-n-2}$. Also, $\val(c,t^A_s)=\val(c,t^A_e)=2-t/2$, for all $A$, because it was true for $v$.

The argument for reverse encoding is similar.
\end{proof}

\section{\NP- and \coNP-hardness}
\label{app:np_conp_hardness}

As a stepping stone towards \PSPACE-hardness, we first show that the problem is
\NP- and \coNP-hard. We will do so by encoding SAT and DNF-tautology, well-known
\NP- and \coNP-hard problems. 

\paragraph{\NP-hardness.}
Given a SAT-formula $F$ (with variables encoded in straight format), we
construct a game $G$ so that a detector state $c$ encodes $F$ (according to the
previous section). 
Our \NP-hardness will use a special state called an {\em extender state}. 

In essence, the job of the extender state $x$ is to extend the duration for
which the formula is true (because of the exists part of SAT - we extend the
duration for which the formula is false for \coNP).  The extender state $x$ is a
maximizer state, has holding rate $1/2$ and an edge to $c$ of cost $0$. Consider some
time $t$. If there is an assignment $A$ such that $t\leq t^A_m$ and $F(A)$, then
it is an optimal strategy in $x$ to wait until $t^A_m$ and then go to $c$.

We will show so by considering all possible strategies in $x$. 
Let $f(t)=v-t/2$.
If we start in $x$ at time $t$, and wait until time $t'$ and then move to $c$ (and afterwards follow some optimal strategy), we get 
\begin{align*}
(t'-t)\cdot \frac{1}{2} + \val(c,t')&=(t'-t)\cdot \frac{1}{2} + \val(c,t')-f(t')+f(t')\\
&=(\val(c,t')-f(t'))+f(t)\enspace ,
\end{align*}
using that $f(t)=v-t/2=v-t/2-t'/2+t'/2=(t'-t)/2+f(t')$.
We thus see that the optimal strategy (for the maximizer, since $x$ belongs to
him) is to wait until a time that maximizes $(\val(c,t')-f(t'))$ (because $f(t)$
is independent of how much we wait). According to Lemma~\ref{lem:detector}, such
times $t'$ are when there is an satisfying assignment $A$ to $F$ for which $t\leq
t^A_m$. If no such satisfying assignment exists, it is an optimal choice to go
to $c$ directly (because we might have that $t$ is just slightly larger than
$t^A_m$ for some satisfying assignment $A$ and the function $(\val(c,t')-f(t'))$
is then decreasing until hitting 0).

Note that $0<t^A_m$ for all assignments $A$ and thus, we have that $\val(x,0)=2+
2^{-n-2}$ if there is a satisfying assignment $A$ to $F$ (and it is optimal to
wait until time $t_m^A$ and then move to $c$ when starting in $x$ at time~0) and
otherwise, if no satisfying assignment exists, $c$ is such that $\val(c,t)=f(t)$
(because it is a straight encoding of the formula), and, by our above
calculations, we see that $\val(x,t)=f(t)$ as well. In particular,
$\val(x,0)=2$. We get the following theorem.

\begin{theorem}
In the SPTG $G$ constructed above $\val(x,0)\in \{2,2+2^{-n-2}\}$ and $\val(x,0)=2+ 2^{-n-2}$ iff there is a satisfying assignment to the SAT instance with boolean formula $F$.
\end{theorem}

Note that the SPTG problem in the theorem is an instance of the PromiseSPTG problem.

Because DecisionSPTG is harder than PromiseSPTG and since we only used holding rates in $\{0,1/2,1\}$, we get the following theorem as a corollary:
%\begin{reptheorem}{thm:np_hard}
%For an SPTG, deciding whether $v(s, 0) \ge c$ for a given state $s$ and constant
%$c$ is \NP-hard, even if the game has only holding rates in $\{0,1/2,1\}$.
%\end{reptheorem}
\nphardthm*

We can show \coNP-hardness similarly, by considering DNF-tautology, using
reverse encoding and having the extender be a minimizer state instead. Changing
the extender to be a minimizer state instead, then ensures that $\val(x,0)\in
\{2,2-2^{-n-2}\}$ and $\val(x,0)=2- 2^{-n-2}$ iff there is an assignment to the
DNF-tautology instance such that it evaluates to false (or in other words, the
formula $F$ is a tautology iff $\val(x,0)=2$).

We get the following theorem:
%\begin{reptheorem}{thm:conp_hard}
%For an SPTG, deciding whether $v(s, 0) \ge c$ for a given state $s$ and constant
%$c$ is \coNP-hard, even if the game has only holding rates in $\{0,1/2,1\}$.
%\end{reptheorem}
\conphardthm*

\section{Quantified boolean formula encoding and \PSPACE-hardness}
\label{app:pspace}

In this section, we will show how to encode a quantified boolean formula in our
encoding. This then trivially allows us to show \PSPACE-hardness, since if we
can encode a quantified boolean formula, we can (easily) solve TQBF. Doing it
this way also allows us to give a more precise picture of what is required for
being hard for the $k$-th level of the polynomial time hierachy.

Consider a quantified booolean formula 
\begin{align*}
\forall v_1^1,\dots, v_{n_1}^1 \ \exists v_{1}^2 \dots, v_{n_2}^2\forall \dots \exists v_{1}^n, \dots, v_{n_n}^n:\\
F(v_{n_1}^1,\dots,v_{n_1}^1,v_{1}^2,\dots,v_{n_n}^n),
\end{align*}
where, like before $F(v_{n_1}^1,\dots,v_{n_1}^1,v_{1}^2,\dots,v_{n_n}^n)$ is assumed to have negations only on the variables (we can still assume so, because of De Morgan's laws).

\paragraph{Variable encoding.}
We use variables $v_1,\dots, v_{n_1}$ to encode $v_1^1,\dots, v_{n_1}^1$, variables $v_{n_1+3},\dots,v_{n_1+n_2+2}$ to encode  $v_1^2,\dots, v_{n_2}^2$ and in general variable $v_{S}$ to encode $v_i^j$, where $S=\sum_{k=1}^{j-1} (n_k+2)+i$.
In particular, we skip two variables whenever we have an alternation. 

\paragraph{Recursive construction.}
We will recursively construct an encoding of our quantified boolean formula by showing:
Given a straight encoding, where $v\geq 2$ and $0<v'\leq 2^{-n-2}$ 
(with free variables $v_1^1,\dots v_{n_1}^1,v_1^2,\dots,v_{n_{i}}^{i}$ - i.e. the ones given by time) of 
\begin{align*}
F_{i+1}^A:=\forall v_1^{i+1},\dots, v_{n_{i+1}}^{i+1}\exists v_1^{i+2},\dots, v_{n_{i+2}}^{i+2}\forall \dots \exists v_{1}^n,\dots, v_{n_n}^n: \\
F(v_{n_1}^1,\dots,v_{n_1}^1,v_{1}^2,\dots,v_{n_n}^n)
\end{align*}
we give a reverse encoding, where $v\geq 2$ and $0<v'\leq 2^{-n-2}$, (with free variables $v_1^1,\dots v_{n_1}^1,v_1^2,\dots,v_{n_{i-1}}^{i-1}$ - i.e. the ones given by time) of 
\begin{align*}
F_{i}^E:=\exists v_1^{i},\dots, v_{n_{i}}^{i}\forall v_1^{i+1},\dots, v_{n_{i+1}}^{i+1}\exists \dots \exists v_{1}^n,\dots, v_{n_n}^n: \\
F(v_{n_1}^1,\dots,v_{n_1}^1,v_{1}^2,\dots,v_{n_n}^n).
\end{align*}

Also, we will show how to, given a reverse encoding, where $v\geq 2$ and $0<v'\leq 2^{-n-2}$ (with free variables $v_1^1,\dots v_{n_1}^1,v_1^2,\dots,v_{n_{i}}^{i}$ - i.e. the ones given by time) of 
\begin{align*}
F_{i+1}^E:=\exists v_1^{i+1},\dots, v_{n_{i+1}}^{i+1}\forall v_1^{i+2},\dots, v_{n_{i+2}}^{i+2}\exists \dots \exists v_{1}^n,\dots, v_{n_n}^n:\\
F(v_{n_1}^1,\dots,v_{n_1}^1,v_{1}^2,\dots,v_{n_n}^n)
\end{align*}
give a straight encoding, where $v\geq 2$ and $0<v'\leq 2^{-n-2}$ (with free variables $v_1^1,\dots v_{n_1}^1,v_1^2,\dots,v_{n_{i-1}}^{i-1}$ - i.e. the ones given by time) of 
\begin{align*}
F_{i}^A:=\forall v_1^{i},\dots, v_{n_{i}}^{i}\exists v_1^{i+1},\dots, v_{n_{i+1}}^{i+1}\forall \dots \exists v_{1}^n,\dots, v_{n_n}^n:\\
F(v_{n_1}^1,\dots,v_{n_1}^1,v_{1}^2,\dots,v_{n_n}^n).
\end{align*}

Note that we can directly give a straight/reverse encoding of $F(v_{n_1}^1,\dots,v_{n_1}^1,v_{1}^2,\dots,v_{n_n}^n)$ using the section on formula encoding, and thus, if we show how to make this recursive construction, we can recursively construct the full quantified boolean formula.

\paragraph{The detector state.}
Let $v_S$ be the variable encoding $v_1^i$.
We construct a state $c$ which encodes the formula $F'=(F_{i+1}^A\wedge v_{S-1}\wedge v_{S-1})$ (i.e. by having states $s_{S-1}^*$ and $s_{S-2}^*$, encoding $v_{S-1}$ and $v_{S-2}$ respectively and then $c$ is a minimizer state of holding rate $1$ with an edge to each of $s,s_{S-1}^*,s_{S-2}^*$. The edge to $s$ has cost 0 and the ones to $s_{S-1}^*,s_{S-2}^*$ have cost $v-2$.
Thus, $\val(c,t)=\min(\val(s,t),\val(s_{S-1}^*,t)+v-2,\val(s_{S-2}^*,t)+v-2)$ (by Lemma~\ref{lem:no wait}). 

\begin{lemma}
The state $c$ straight encodes $F'$ with the same $v$ and $v'$ as for $s$. 
\end{lemma}
\begin{proof}
Consider time split into durations of length $2^{-S+3}$ (i.e. a duration for $v_{S-3}$) each. Then, for any time $t$ in the first $3/4$ of the duration, we have that $\val(c,t)=v-t/2$, because the assignment $A$ for which $t\in T^A$ is such that at least one of $v_{S-1}$ and $v_{S-2}$ is false. Still, the state $c$ is a straight encoding of $(F\wedge v_{S-1}\wedge v_{S-1})$ (with the same $v$ and $v'$ as for $s$): Property (str2) comes, similar to earlier cases, from that each of the functions $\val(s,t),\val(s_{S-1}^*,t)+v-2$ and $\val(s_{S-2}^*,t)+v-2$ satisfies it. Property (str1) can be seen since $v-t/2\leq \val(s,t),\val(s_{S-1}^*,t)+v-2,\val(s_{S-2}^*,t)+v-2$ and thus also the minimum of them. Also, $\val(c,t)\leq \val(s,t)\leq v+v'-t/2$. That (str3) is satisfied comes from that for the last $1/4$ of the durations of length $2^{-S+3}$, we have that $\val(s_{S-1},t)$ starts with $v-t/2$ and then increases with the highest holding rate in the game (i.e. holding rate 1) until $7/8$ of the duration and then falls back to $v-t/2$ at the very end with the smallest holding rate in the game (i.e holding rate 0). Therefore, $\val(s_{S-1},t)\geq \val(s,t)$ for such $t$. Similarly, $\val(s_{S-1},t)$ increases, from value $v-t/2$, with the highest holding rate from the middle of these durations, until $3/4$ into the duration at which point it falls back to $v-t/2$ at the end. Again, therefore $\val(s_{S-2},t)\geq \val(s,t)$. In particular, $\val(c,t)=\val(s,t)$ for such $t$ and thus, since $\val(s,t)$ satisfies (str3) for such $t$, it is also satisfied by $c$. On the other hand, in the first 3/4 of these durations, at least one of $\val(s_{S-1},t)$ and $\val(s_{S-2},t)$ have value $v-t/2$, because they are straight encodings (and all such $t$ belongs to $T^A$ for some assignnment where either $v_{S-1}$ or $v_{S-2}$ are false).
\end{proof}

\paragraph{The extender state.}
Next, let $x$ be a maximizer state with holding rate $r(x)=1/2-\frac{v'}{2.5\cdot 2^{-S+1}}$ and an edge to $c$ of cost $0$.
Intuitively, $x$ is an extender, similar to in our \NP-hardness proof, for extending the truth value of our assignments to all variables, but which resets (back to having value $v-t/2$) between each assignment of the first $S-3$ variables.

Consider (like in the proof for $c$ being a straight encoding) time split into durations of length $2^{-S+3}=4\cdot 2^{-S+1}$ (i.e. a duration for $v_{S-3}$) each. 
Consider such a duration, which corresponds to some assignment $\hat{A}$ of variables $v_1^1,\dots,v_{S-3}$.

\begin{lemma}\label{lem:slim region}
For all $t$ in a slim region, from $11/16$ to $3/4=12/16$ into the duration, either $\val(x,t)\geq v+v'/2-t/2$, if $F_i^E(\hat{A})$ is true, or $\val(x,t)=v-t/2$ if $F_i^E(\hat{A})$ is false
\end{lemma}
\begin{proof}
We will first argue that the maximizer never wait $w>2.5\cdot 2^{-S+1}$ in $x$ from time $t$, for any $t$. 
This is because the outcome is then 

\begin{align*}
r(x) \cdot w+\val(c,t+w)\leq & \ \frac{w}{2} - w\cdot \frac{v'}{2.5\cdot 2^{-S+1}}+ v+v'-\frac{t+w}{2}\\
< & \ v- \frac{t}{2}.
\end{align*}
On the other hand not waiting at all gives an outcome of $\val(c,t)\geq v-t/2$, because $c$ is a straight encoding. 

Observe that $\val(c,t)=v-t/2$ for $t$ from $0$ to $3/4$ into a duration, because at least one of $v_{S-1}$ and $v_{S-2}$ are false.
Let $t'$ be some number between $11/16$ and $3/4$ into the duration. If for all $t$ the durations, $\val(c,t)=v-t/2$, then the maximizer will not wait in $x$ in that duration, because, he cannot, as shown above, wait until $3/4$ or later into the next duration. This means that $\val(c,t)=\val(x,t)=v-t/2$ for such durations. In particular, $\val(x,t')=v-t'/2$.
But if $\val(c,t)=v-t/2$ for each $t\in T^A$ for each assignment $A$ that extends $\hat{A}$ into an assignment to variables $v_1^1,\dots, v_{n_i}^i$, then, $F_i^E(\hat{A})$ is false.

On the other hand, if for some $t$ in a duration, we have that
$\val(c,t)=v+v'-t/2$ (which is the case if $\val(c,t)\neq v-t/2$ for all $t$ in
such a duration, by straight encoding), then, $\val(x,t')\geq v+v'/2-t/2$,
because, $t>t'$ (since $\val(c.t'')=v-t/2$ for all $t''$ before 3/4th into the
duration) and thus $t-t'$ is at most 5/16 of the duration. If we wait $d\leq
w=5/16\cdot 2^{-S+3}=5/4\cdot 2^{-S+1}$ from time $t$ until time $t'$, we get an
outcome of 
\begin{alignat*}{2}
r(x)\cdot d+\val(c,t') ={}& r(x)\cdot d+v+v'-t'/2  \\
					   ={}& d/2-\frac{d\cdot v'}{2.5\cdot 2^{-S+1}}+v+v'-t'/2 \\
					   ={}& (v'-\frac{d\cdot v'}{2.5\cdot	2^{-S+1}})+v-t/2\\
					   \geq{}& (v'-\frac{w\cdot v'}{2.5\cdot 2^{-S+1}})+v-t/2\\
					   ={}& v+v'/2-t/2.
\end{alignat*}
But if $\val(c,t)=v+v'-t/2$ for some $t\in T^A$ and some assignment $A$ that extends $\hat{A}$ into an assignment to variables $v_1^1,\dots, v_{n_i}^i$, then, $F_i^E(\hat{A})$ is true.
\end{proof}

\paragraph{Limiter.}
We will now, using a few more states, use the above lemma to create a state $s'$
that is a reverse encoding of $F_i^E$ with $v_i\leftarrow v+v'/2$ and
$v_i'\leftarrow v'/2$.
First, to ensure that $\val(s',t)\leq v_i-t/2$, we will introduce a limiter
state $L$. The state $L$ is a minimizer state of holding rate $1$, with an edge to $x$
and one to a maximizer state $s''$ of holding rate $1/2$, each of cost $0$. The state
$s''$ has a single edge to a goal state of cost $v_i-1/2$. Thus,
$\val(s'',t)=v_i-1/2+(1-t)/2=v_i-t/2$ (because the maximizer will wait until
time 1 and then move to goal).  
We have that $\val(L,t)=\min(\val(x,t),\val(s'',t))=\min(\val(x,t),v_i-t/2)$, by Lemma~\ref{lem:no wait}.

\paragraph{The reverse encoding state $s'$.}
The state $s'$, which we will show, reverse encodes $F_i^E$ is a maximizer state
of holding rate $0$ that has an edge to $L$ and to a state $r$ reverse encoding
$F''=(\neg v_{S+1}\vee \neg v_{S}\vee v_{S-1}\vee \neg v_{S-2})$ with $v_i$ and
$v_i'$.

\begin{lemma}
The state $s'$ reverse encodes $F_i^E$ with $v_i$ and $v_i'$
\end{lemma}
\begin{proof}
We have that $\val(s',t)=\max(\val(L,t),\val(r,t))$ by Lemma~\ref{lem:no wait}.

First, property (rev1) follows from that $\val(L,t)\leq v_i-t/2$ and $\val(r,t)\in [v_i-t/2,v_i-v_i'-t/2]$ (by reverse encoding), and thus, $\val(s',t)=\max(\val(L,t),\val(r,t))\in [v_i-t/2,v_i-v_i'-t/2]$.

For property (rev2) and (rev3), consider some assignment $\hat{A}$ to the variables $v_1,\dots, v_{S-3}$. 
For each $t\in \{t^{\hat{A}}_s,t^{\hat{A}}_e\}$, we have that $\val(s',t)=\val(s',t)=v_i-t/2$, because $\val(L,t)=v_i-t/2$ for such $t$ (by reverse encoding and since such $t$ are also the end of durations for higher numbered variables) and $\val(s',t)\leq v_i-t/2$ and thus, $\val(s',t)=\max(\val(L,t),\val(r,t))=v_i-t/2$.

 We have that the duration for $\hat{s}_{S+1}^*$, the state encoding $\neg v_{S+1}$, has a duration of $2^{-S-1}=2^{-S+3}/16$. Also, the boolean encoding of $11/16$ is $0.1011$, and thus, $F''$ is true, except for a period of length $1/16$, starting at $11/16$ into the duration. For any $t\not\in [11/16,3/4]$, we have that (1)~$\val(r,t)=v_i-t/2$; (2)~$\val(L,t)\leq v_i-t/2$  and (3)~$\val(s',t)=\max(\val(L,t),\val(r,t))$ (by Lemma~\ref{lem:no wait}), we have that $\val(s',t)=v-t/2$. 
For any $t\in [11/16,3/4]$, we have by Lemma~\ref{lem:slim region} that $\val(c,t)$ is greater than and therefore $\val(L,t)$ is equal to $v+v'/2-t/2=v_i-t/2$ if $F_i^E(\hat{A})$ is true and equal to $v-t/2=v_i-v_i'-t/2$ if $F_i^E(\hat{A})$ is false. 
Let $A$ be the extension of $\hat{A}$ that maps $v_{S+1},v_S, v_{S-2}$ to true and $v_{S-1}$ to false.
Note that $T^A$ is the duration covering $[11/16,3/4]$ of the duration. 
We have that $\val(L,t^A_m)=v_i-v_i'-t/2$, by it reverse encoding $F''$ with those parameters.
Thus, for all $t$ in the duration, $\val(s',t)=v_i-t/2$ if $F_i^E(\hat{A})$ is true and otherwise, $\val(s',t^A_m)=v_i-v_i'-t/2$ if $F_i^E(\hat{A})$ is false.
\end{proof}

We can similar give a straight encoding of $F_{i}^A$, with $v\leftarrow v-v'/2$ and $v'\leftarrow v'/2$ if given a reverse encoding of $F_{i+1}^E$ with $v$ and $v'$.

We can thus encode any quantified boolean formula. In particular, we see that PromiseSPTG is \PSPACE-hard. 
Hence, we get the following theorem.

%\begin{reptheorem}{thm:pspace}
%For an SPTG, deciding whether $v(s, 0) \ge c$ for a given state $s$ and constant
%$c$ is \PSPACE-hard.
%\end{reptheorem}
\pspacethm*

Also, since we need holding rates $\{0,1\}$ for our family with exponentially many event points and besides that one more distinct holding rate (for the extender state) per alternation to encode a quantified boolean formula, we also get the following theorem.

%\begin{reptheorem}{thm:poly_hier} 
%For an SPTG with $k+2$ distinct holding rates, deciding whether $v(s, 0) \ge c$
%for a given state $s$ and constant $c$ is hard for the $k$-th level of the
%polynomial time hierarchy.
%\end{reptheorem}
\polyhierthm*

\section{Graph properties}
\label{app:graph_prop}
We will in this section argue that our construction (or similar constructions) belongs to very many special cases of graphs.
A simple way to understand our \PSPACE-hardness construction is that it basically forms a tree, but with the leaves being a member of our family of graphs with exponentially many event points.

First, we will give three simple properties.
\begin{enumerate}[leftmargin=0.4cm]
\itemsep2mm
\item {\bf DAG.} Note that our lower bound graph is acyclic, since our family of graphs with exponentially many event points were, and it is thus a DAG
\item {\bf Planar.} It is also planar (i.e. it can be drawn in the plane without edges crossing), since we can draw our family of graphs with exponentially many event points in a planar way such that the last level (i.e. states $s_{\ell}^i,s_{r}^i$) are on the outmost surface and thus, a tree, like our \PSPACE-hardness construction is, with such leaves are planar. %reformulate?
\item {\bf Degree 3 or 4.} By assuming that ANDs and ORs are over two variables (which we can do without loss of generality), we see that our graph has degree at most 4. In fact, we can get it down to 3, by for each state $s$ (1)~adding a maximizer state $s'$ with holding rate 0; (2)~adding an edge from $s'$ to $s$ of cost $0$, thus ensuring that $\val(s',t)=\val(s,t)$ by Lemma~\ref{lem:no wait} and (3)~letting all incoming edges from $s$ go to $s'$ instead. The resulting game is such that each state has the same value as before the modification, but the degree is at most~3. Doing so does not change any of other graph properties.
\end{enumerate}

\paragraph{Few holding rates.}
 First, observe that we only use two holding rates, 0 and 1, in our exponential lower bound family. Also, in our \NP-hardness construction we only use holding rates $0,1/2,1$. In general, for a formula with $k-1$ alterations, we use $k+2$ holding rates (2 for the exponential lower bound and 1 more for each $\exists$ or $\forall$ in the formula) and hence SPTGs with $k+2$ holding rates are hard for the $k$-th level of the polynomial time hierachy.
 
\subsection{Urgent states}
One can consider a slightly more general variant of OCPTGs and SPTGs where a state may be declared {\em urgent}. In an urgent state, the owner cannot wait. 

The notion is of interest at least partly, because the proofs in \cite{R11,BLMR06}, showing that OCPTGs have a value and at most exponentially many event points respectively, are recursive arguments on graphs with more and more urgent states. 

We will say that a set of states $S$ is {\em urgent-equivalent} if, for some pair of optimal strategies, no player waits in any state of $S$. Note that changing any number of states in an urgent-equivalent set to being urgent does not change the value (because some optimal strategy in the original game is still useable).
Observe that our family with exponentially many event points are such that $S$ is urgent-equivalent, where $S$ is all states, except for $s_r^0$, because of Lemma~\ref{lem:no wait}.
Hence, even in games with all but 1 state being urgent, we have exponentially many event points.

Next, consider that in our \NP-hardness or our \PSPACE-hardness construction, 
we reuse the state $s_r^0$ also use it instead of state $L$ (used to bound booleans - the two states have the same value function) instead of coping it for each time we use our family in the construction. Then, similar to above, we see that the problem is \NP-hard with 2 non-urgent states (the state $s_r^0$ and the extender state) and in general, with $k+1$ non-urgent states, it is hard for the $k$-th level of the polynomial time hierarchy (we add one extender state for each alternation). Note that the resulting graph is not planar (and also does not have low degree, but it can still be changed similar to before to have degree 3 though while still satisfying all properties except for planarity). One can also show that it has tree-width/clique-width 4 as defined later (by simply having $s_r^0$ in all bags and giving it a unique color respectively). Since it has clique-width 4 it has rank-width 4 as well.

\subsection{Treewidth}
Treewidth is a classic measure of how tree-like a graph is, with treewidth being 1 if the graph is a tree. Many algoritmic problems are easier on graphs with constant tree-width, e.g. monodic second order logic, which is lucky since many graph families (e.g. control flow graphs of C programs without gotos, used in verification of programs) have constant tree-width.
We will argue that our \PSPACE-hard family have tree-width 3, thus showing that solving SPTGs on constant tree-width graphs are \PSPACE-hard.

\paragraph{Treewidth definition.}
A {\em tree-decomposition} of a graph $G=(V,E)$ is a pair $(B,T)$ where $B=\{B_1,B_2,\dots,B_k\}$ is a family of subsets of $V$ (called {\em bags}), and $T$ is a tree with the bags as states, satisfying that 

\begin{itemize}[leftmargin=0.3cm]
\itemsep2mm
\item for each edge $(s,t)\in E$, there exists $i$, such that $s,t\in B_i$. 

\item for each state $s\in V$, the set of bags $X_s$ containing $s$ (i.e. $B_i\in X_s$ iff $s\in B_i$), is a non-empty subtree in $T$. 

\end{itemize}

The {\em width} of a tree-decomposition is $\max_i |B_i|-1$.
The {\em tree-width} of a graph $G=(V,E)$ is the minimum width of any tree-decomposition of $G$ (there can be many tree-decompositions of the same graph).
The {\em path-width} of a graph $G=(V,E)$ is the minimum width of any tree-decomposition $(X,T)$ of $G$, for which $T$ is a path.

\paragraph{Tree-decomposition construction.}
We will first argue that our exponential lower bound example have path-width 3.

The bags are such that $B_k=\{s_{\ell}^{k-1},s_{r}^{k-1},s_{\ell}^k,s_{r}^k\}$ for $k\in \{1,\dots,i\}$. 
The bag $B_k$ has an edge to $B_{k+1}$ for $k<i$ and to $B_{k-1}$ for $k>1$.

Note that all edges are in some bag and that each state is in (at most three) consecutive bags.

Next, the full graph has tree-width 3 as well, simply because it is a tree where the leaves has an edge to a a state in the top bag of a graph with pathwidth (and therefore tree-width) 3.

\subsection{Clique-width}
There are more general graph properties than tree-width (even if tree-width is the most well-known) for which many algorithms runs faster on the special case. One of these are clique-width. It is known that if the treewidth is $w$ then clique-width is at most $3\cdot 2^{w-1}$. In our case, because the tree-width was 3, that would imply that the clique-width is at most 12. We will argue that it is in fact 3.

\paragraph{Definition of clique-width.}
We will next define {\em clique-width}. To do so, let $k$ be some number. The
set of graphs of clique-width at most $k$ are exactly those that can be
constructed using the following set of rules: 

\smallskip

\begin{enumerate}[leftmargin=0.4cm]
\itemsep1mm
\item Given an integer $1\leq i\leq k$, output a new graph with one state of color $i$

\item Given two graphs $G,G'$, output the disjoint union of them

\item Given a graph and two integers $i\neq j$ such that $1\leq i,j\leq k$, output the graph you get by adding an edge from each state of color $i$ to each state of color $j$

\item Given a graph and two integers $i\neq j$ such that $1\leq i,j\leq k$, output the graph you get by changing the color of each state of color $i$ to color $j$.

\end{enumerate}

\paragraph{Lower bound on clique-width.}
A graph $H=(V',E')$ is an {\em induced sub-graph} of a graph $G=(V,E)$, if there exists an injective function $f:V'\rightarrow V$, such that for all pairs $s,t\in V'$, we have that $(s,t)\in E'$ iff $(f(s),f(t))\in E$.
There are also futher types, e.g. rank-width, however, if a graph has clique-width $k$, then it has rankwidth at most $k$ as well and thus our graph has rank-width 3.

It is known that the set of graphs  (with edges) that have clique-width 2 are exactly the ones that do not have a path of length 4 as an induced sub-graph.

Note that already our exponential lower bound family (with $i\geq 3$) has an induced sub-graph which is a path of length 4. E.g. $s_{\ell}^0,s_{\ell}^1,s_{\ell}^2,s_{\ell}^3$ is such an induced sub-graph. Thus, the clique-width is at least 3.

\paragraph{Upper bound on clique-width.}
We first construct our family with exponentially many event points. 

We do so in steps. In the first few steps, we construct $s_{\ell}^0$ and $s_{r}^0$, add the edge between them (this uses 2 colors) and color them, say, green.

We then do an iterative construction to construct all $i$ levels as follows. The last level is special though, in that only one state of that level has incoming edges (and we could have omitted the other state from our construction). Consider that for some $k\leq i-2$, states $s_{\ell}^k$ and $s_r^k$ are green, states $s_{\ell}^{j},s_r^j$, for $0\leq j<k$ are blue and all edges between them have been added, then, we can add $s_{\ell}^{k+1}$ and $s_{r}^{k+1}$ as a red state, add all edges between red and green states and then color green states blue and then red states green and we go to the next iteration. We will only use one of $s_{\ell}^{i}$ and $s_r^i$ in the tree above our family, so add first the other (i.e. if we want to use  $s_{\ell}^{i}$, add first $s_r^i$) as a red state, add all edges between red and green, color the red state blue and then add the one we want to use later as a red state, add all edges between red and green and color the green states blue and finally the red state green. This used 3 colors.

To construct our full \PSPACE-hardness construction, we construct each family member as above and then construct the graph above level by level (i.e. add the next higher state in the tree as a red state, add all edges between red and green and recolor the green states blue and then the red state green and proceed upward like that).
Again, this uses 3 colors in total.

\section{Integer costs}
\label{app:ashutosh}

Here we show how to convert our games with exponentially-many event points
and two holding rates to have integer costs. 

Fix some $i$. Taking the member with $i$ levels our family
of SPTGs that have an exponential number of event points, and changing the unit
of time to be $2^{-i}$ of the old time unit (thus, the duration of the game that
was previously of length 1 time unit is now $2^i$ time units) and changing the
currency for the output to be $2^{-i}$-th of the old outputs currency, we see
that the costs on edges are scaled with $2^i$ (and are thus integers), the holding rates
are scaled with $2^i\cdot 2^{-i}=1$ (and are thus still in $\{0,1\}$), because
the change in time unit and in currency cancels and finally the duration has
changed to $2^i$ time units, i.e. time is in $[0,2^i]$. The resulting game then
has the same number of event points ($2^{i+1}$), but have integer costs and
holding rates in $\{0,1\}$.

\end{document}